\documentclass[10pt,conference]{IEEEtran}
\usepackage[english]{babel}
\usepackage[T1]{fontenc}    
\usepackage{cite}
\usepackage[bookmarks=false,draft]{hyperref} 
\usepackage{amsfonts}       
\usepackage[pdftex]{graphicx}
\usepackage{amsthm,array,booktabs,tabularx,caption}
\usepackage{makecell}
\usepackage[cmintegrals]{newtxmath}
\usepackage{bm} 

\usepackage{algorithm}
\usepackage{algorithmic}
\usepackage{stfloats}
\usepackage{xcolor}
\usepackage{soul}
\usepackage{etoolbox}
\usepackage{url}
\usepackage{paralist}
\usepackage{enumitem}
\usepackage{subcaption}
\theoremstyle{definition}
\newtheorem{thm}{Theorem}[section]

\newtheorem{rem}{Remark}[section]
\newtheorem{cor}{Corollary}[thm]

\theoremstyle{definition}
\newtheorem{defi}{Definition}[section]

\makeatletter
\def\ps@headings{
\let\@oddhead\@empty
\let\@evenhead\@empty
\def\@oddfoot{\@IEEEheaderstyle\hfil\thepage}%
\def\@evenfoot{\@IEEEheaderstyle\thepage\hfil\hbox{}}
}
\def\ps@IEEEtitlepagestyle{
\let\@oddhead\@empty
\let\@evenhead\@empty
\def\@oddfoot{\mycopyrightnotice}%
\let\@evenfoot\@empty
}
\makeatother

\def\mycopyrightnotice{
  {\footnotesize
  \begin{minipage}{\textwidth}
Accepted by INFOCOM \copyright2022 IEEE. Personal use of this material is permitted. Permission from IEEE must be obtained for all other uses, in any current or future media, including reprinting/republishing this material for advertising or promotional purposes, creating new collective works, for resale or redistribution to servers or lists, or reuse of any copyrighted component of this work in other works.
  \end{minipage}
  }
}

\begin{document}
\title{Multi-Agent Distributed Reinforcement Learning for Making Decentralized Offloading Decisions}

\author{\IEEEauthorblockN{Jing Tan}
\IEEEauthorblockA{\textit{Huawei Technology} \\
Munich, Germany \\
jingtan@huawei.com}
\and
\IEEEauthorblockN{Ramin Khalili}
\IEEEauthorblockA{\textit{Huawei Technology} \\
Munich, Germany \\
ramin.khalili@huawei.com}
\and
\IEEEauthorblockN{Holger Karl}
\IEEEauthorblockA{\textit{Hasso Plattner Institute} \\
Berlin, Germany \\
holger.karl@hpi.de}
\and
\IEEEauthorblockN{Artur Hecker}
\IEEEauthorblockA{\textit{Huawei Technology} \\
Munich, Germany \\
artur.hecker@huawei.com}
}

\maketitle

\begin{abstract}
We formulate computation offloading as a decentralized decision-making problem with autonomous agents. We design an interaction mechanism that incentivizes agents to align private and system goals by balancing between competition and cooperation. The mechanism provably has Nash equilibria with optimal resource allocation in the static case. For a dynamic environment, we propose a novel multi-agent online learning algorithm that learns with partial, delayed and noisy state information, and a reward signal that reduces information need to a great extent. Empirical results confirm that through learning, agents significantly improve both system and individual performance, e.g., $40\%$ offloading failure rate reduction, $32\%$ communication overhead reduction, up to $38\%$ computation resource savings in low contention, $18\%$ utilization increase with reduced load variation in high contention, and improvement in fairness. Results also confirm the algorithm's good convergence and generalization property in significantly different environments.
\end{abstract}

\section{Introduction}
\label{sec:intro}

Vehicular network (V2X) applications are characterized by huge number of users, dynamic nature, and diverse Quality of Service (QoS) requirements \cite{masmoudi2019survey}. They are also computation-intensive, e.g., inferring from large neural networks \cite{hofmarcher2019visual} or solving non-convex optimization problems \cite{claussmann2019review}\cite{badue2020self}. These applications currently reside in the vehicle's onboard units (OBU) for short latency and low communication overhead. Even with companies such as NVidia developing OBUs with high computation power \cite{oh2019hardware}, post-production OBU upgrades are typically not commercially viable; and irrespective of local OBU power, the ability to offload tasks to edge/cloud via multi-access edge computing (MEC) devices increases flexbility, protecting vehicles against IT obsolescence. Hence, offloading is a key technique for future V2X scenarios \cite{europe6g,you2021towards,5gaausecase1,5gaausecase2}. 


Currently, computation offloading decisions are strictly separated between the user and the operating side \cite{mach2017mobile}. Users decide what to offload to optimize an individual goal, e.g., latency \cite{baidya2020vehicular}
or energy efficiency \cite{loukas2017computation}. Apart from expressing their preference through a pre-defined, static and universal QoS matrix \cite{masdari2021qos}, users cannot influence how their tasks are prioritized. The operating side centrally prioritizes tasks and decides resource allocation to optimize a system goal that is based on the QoS matrix, but not always the same as the users' goals, e.g., task maximization \cite{choo2018optimal} or load-balancing \cite{vondra2014qos}. 

This separation poses problems for both user and operating sides, especially in the V2X context. V2X users have private goals \cite{shivshankar2014evolutionary}, are highly autonomous \cite{martinez2010assessing}, reluctant to share information or cooperate, and disobedient to a central planner \cite{feigenbaum2007distributed}. They want flexible task prioritization and influence on resource allocation without sharing private information \cite{li2019learning}. On the operating side, edge cloud computing architecture introduces signalling overhead and information delay in updating site utilization \cite{mach2017mobile}; coupled with growing user autonomy and service customization, traditional centralized optimization methods for resource allocation become challenging due to unavailability of real-time information and computational intractability.

Nonetheless, efforts are made on the operating side to apply centralized methods under such conditions, e.g., using heuristics at run-time \cite{kuo2018deploying} or decoupling into smaller problems \cite{agarwal2018joint}, but these solutions still assume complete information. Other efforts are made to jointly optimize private and system goals through game theoretic approaches---although they naturally deal with \textit{decentralized incentives}, they often require complete information of the game to \textit{centrally execute} the desired outcome. E.g., both \cite{chen2014decentralized} and \cite{cardellini2016game} model network resource allocation problems of autonomous users as a game, but \cite{chen2014decentralized} assumes users share decision information, and \cite{cardellini2016game} assumes all user and node profiles are known \textit{a priori}. None of these assumptions are plausible in practice. 

We, hence, need an interaction mechanism between user and operating sides based on incentives, not rules, and an algorithm that makes decentralized decisions with partial and delayed information in a dynamic environment. There are several challenges with such a mechanism. Users may game the system, resulting in potentially worse overall and individual outcomes \cite{oh2008few}---the first challenge \textbf{C1} is how to incentivize user behavior such that users willingly align their private goals to the system goal while preserving their autonomy. The second challenge \textbf{C2} is finding an algorithm that efficiently learns from partial information with just enough feedback signals, keeping information-sharing at a minimum.

Among learning algorithms for decentralized decision-making, no-regret algorithms apply to a wide range of problems and converge fast; however, they require the knowledge of best strategies that are typically assumed to be static \cite{chang2007no}. Best-response algorithms search for best responses to other users' strategies, not for a equilibrium---they are therefore adaptive to a dynamic environment, but they may not converge at all \cite{weinberg2004best}. To improve the convergence property of best-response algorithms, \cite{bowling2002multiagent} introduces an algorithm with varying learning rate depending on the reward; \cite{weinberg2004best} extends the work to non-stationary environments. However, both these algorithms provably converge only with restricted classes of games. The challenge \textbf{C3} still exists to trade off between optimality and convergence, while keeping computation and communication complexity tractable \cite{feigenbaum2007distributed}. 

We propose a decentralized decision-making mechanism based on second-price sealed-bid auctions that successfully addresses these challenges, using the V2X context as an example. Our method is not restricted to V2X applications---it can be applied to other applications facing similar challenges. 

Second-price auctions are commonly used to distribute public goods, due to its welfare-maximization property. Typically, in a second-price sealed-bid auction, a bidder has no knowledge of other bidders' bidding prices and it only receives the bidding outcome as feedback signal. Additionally, it receives the final price if it wins the bid---this befits our requirement to limit information-sharing. Our mechanism also utilizes the feedback signal to incentivize cooperative behavior and speed up learning.
We prove that in the static case, the outcome of this mechanism is a Nash equilibrium (NE) and a maximization of social welfare; under specific conditions (Sec.\ref{highContention}), it is also a \emph{Pareto-optimal} allocation of resources (\textbf{C1}). 

For the dynamic case, we choose to use a multi-agent reinforcement learning (MARL) algorithm, for its ability to learn with partial, noisy and delayed information, and a single reward signal (\textbf{C2}). 
Specifically, our core RL algorithm learns the best-response strategy updated in a fictitious self play (FSP) method. FSP addresses strategic users' adaptiveness in a dynamic environment by evaluating state information incrementally, and by keeping a weighted historical record \cite{heinrich2015fictitious}; it is easier to implement than the method proposed in \cite{bowling2002multiagent}, especially with a large state and action space (\textbf{C3}). Our empirical results show that over time, the best-response strategies stabilize and lead to significantly improved individual and overall outcomes. We compare active (learning-capable) and passive (learning-incapable) agents 
in both synthetic and realistic V2X setups. Our algorithm demonstrates capability to generalize to very different, unseen environments.

\begin{figure*}[t]
	\centering
	\begin{minipage}{0.95\linewidth}
	\subcaptionbox{Topology \label{topo}}{\includegraphics[width=0.24\linewidth]{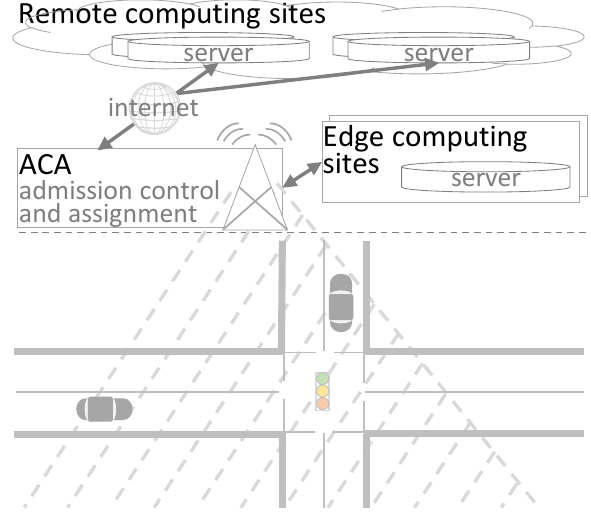}}
	\subcaptionbox{Message sequence \label{flow}}{\includegraphics[width=0.72\linewidth]{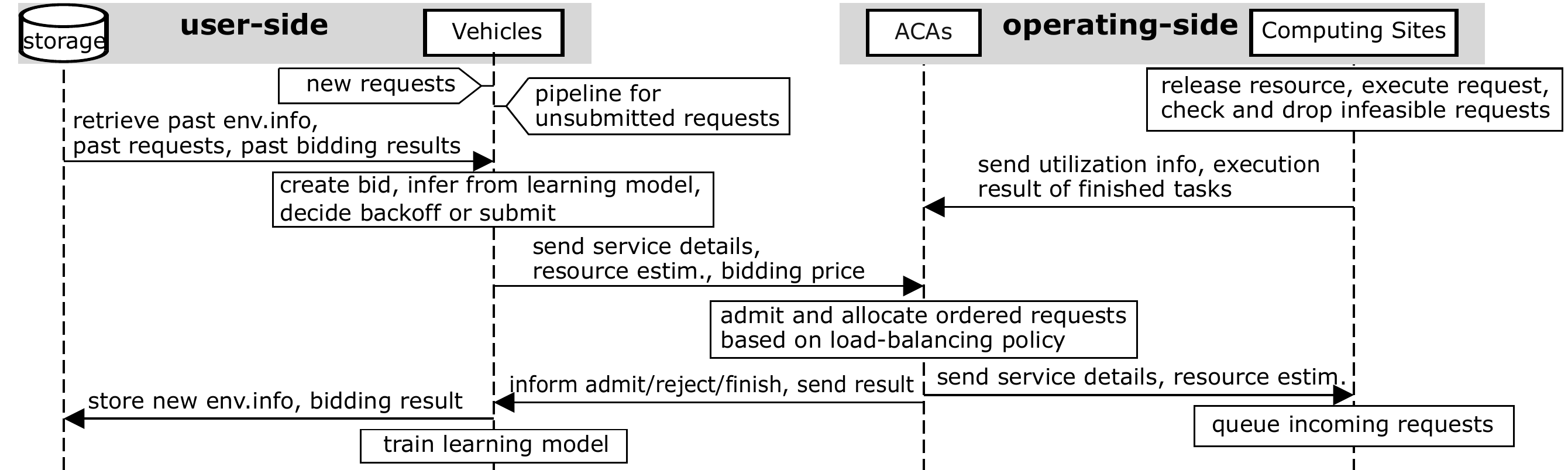}}
	\end{minipage}
	\vspace{-0.2cm}
	\caption{System model}
	\vspace{-0.8cm}
	\label{model}
\end{figure*}

To summarize, our main contributions are:
\begin{itemize}
\item We formulate computation offloading as a decision-making problem with decentralized incentive and execution. The strategic players are incentivized to align private and system goals by balancing between competition and cooperation.
\item We introduce DRACO, a distributed algorithm that learns based on delayed and noisy environment information and a single reward signal. Our solution reduces the requirement for information-sharing to a great extent.
\item We evaluate DRACO in a synthetic setup with randomized parameters, as well as in a realistic setup based on specific mobility model and self-driving applications. Our results show that it significantly increases resource utilization, reduces offloading failure rate, load variation and communication overhead, even in a dynamic environment where information-sharing is limited. The models are easily generalized in different environments. 
\item The authors have provided public access to their code or data at \cite{dracosource}.

\end{itemize}

Sec.\ref{sec:modelproblem} describes system model and problem formulation, Sec.\ref{sec:solution} proposes our solution, Sec.\ref{sec:eval} presents empirical results, Sec.\ref{sec:related} summarizes related work, Sec.\ref{sec:conclusion} concludes the paper.

\section{System Model and Problem Formulation}
\label{sec:modelproblem}

\subsection{System model}
\label{sec:model}

Our system adopts the classic edge cloud computing architecture: user-side vehicles request services; operating-side admission control and assignment (ACA) units (e.g., road side units or base station) control admission of service requests and assign them to different computing sites, which own resources and execute services~\cite{whaiduzzaman2014survey} (Fig.\ref{topo}). 
We propose changes only to: \begin{inparaenum}[1)] \item the algorithm deciding admission and assignments, \item the interaction mechanism. \end{inparaenum} In addition, most signalling needs in our proposed approach are covered by the ISO 20078 standard on extended vehicle web services \cite{iso20078}; additional fields required to pass bidding information are straightforward to implement. Channel security is not the focus of this study.

We first define \emph{service request} in our study; then we explain in detail the user side and the operating side.

\subsubsection{Service request as bid} 

The cloud-native paradigm decomposes services into tasks that can be sequentially deployed~\cite{alliance2019service}. A service request comprises \begin{inparaenum}[1)] \item a task chain, with varying number, type, order and resource needs of tasks, and \item a deadline. \end{inparaenum} We consider a system with custom-tailored services placed at different computing sites in the network; the properties of these services are initially unknown to the computing sites. This enables us to extend use cases into new areas, e.g., self-driving~\cite{5gaausecase1}\cite{5gaausecase2}. We consider independent services, e.g., in self-driving, segmentation and motion planning services can be requested independently.

Classic decentralized decision-making mechanisms include dynamic pricing, negotiations, and auctions. Among these mechanisms, auction is most suitable in a dynamic and competitive environment, where the number of bidders and their preferences vary over time and distribution of private valuations is dispersed \cite{schindler2011pricing}\cite{einav2018auctions}. Among various forms of auction, second-price seald-bid auction maximizes welfare rather than revenue and has limited information-sharing, hence befitting the requirements in our study. Specifically, our approach is based on Vickrey-Clarke-Groves (VCG) for second-price combinatorial auction \cite{vickrey1961counterspeculation}. In our case, we use simultaneous combinatorial auction as a simplified version of VCG---each bidder bids for all commodities separately and simultaneously, without having to specify its preference for any bundle \cite{feldman2013simultaneous}. Since it assumes no correlation between commodities, the simplification befits our study of independent service requests.

We conceive of a vehicle's service request as a \emph{bid} in an auction. Besides the service details, a bid includes the bidding price and the vehicle's estimated resource needs.

\subsubsection{User side} 

We focus on the behavior of vehicles, conceived of as \emph{agents}. A vehicle acts autonomously and privately: it shares no information with other vehicles and only very limited information with the ACA (Sec.\ref{payment} and \ref{fsp}). Its decision objective is to maximize average utility from joining the auction. If it bids a low price and loses, it suffers costs including transmission delay and communication overhead for bidding and rebidding; if it bids a high price and wins, it has reduced payoff. For lower cost or better payoff, it can decide to join the auction at a later time (i.e.\ backoff \cite{cali2000ieee}). But if backoff is too long, it has pressure to pay more to prioritize its request. Therefore, it balances between two options: i) back off and try later or ii) submit the bid immediately to ACA. In our approach, \begin{inparaenum}[1)] \item vehicles are incentivized to balance between backoff and bidding through a cost factor; \item backoff time is learned from state information, not randomly chosen. \end{inparaenum}

We study the learning algorithm in each vehicle. We use passive, non-learning vehicles as benchmark, to quantify the effect of learning on performance. Learning essentially sets the priority of a service request, this priority is used by the ACA to order requests; it is simply constant for non-learning agents, resulting in first-in, first-out processing order.

\subsubsection{Operating side} 

The ACA unit and computing sites are the operating side (Fig.\ref{flow}). The ACA unit decides to admit or reject ordered service requests. Upon admission, it assigns the request to a computing site according to a load-balancing policy. 
Due to information delay, execution uncertainty, system noise, etc., the resource utilization information at different sites is not immediately available to the ACA unit. If all computing sites are overloaded, service requests are rejected, and vehicles can rebid for a maximum number of times. If the request is admitted but cannot be executed before deadline, the computing site drops the service and informs ACA unit. Vehicles receive feedback on bidding and execution outcome, payment, and resource utilization (Sec.\ref{payment}).

The operating side does not have \emph{a priori} knowledge of the type, priority, or resource requirement of service requests. For example, if at run-time, a site receives a previously unknown service, it uses an estimate of resource needs provided by the vehicle. Over time, a site updates this estimate from repeated executions of the same service. Extension to a more sophisticated form of learning is left to future work.

The total service time of a request is the sum of processing, queueing, and transmission time. Each computing site may offer all services but with different resource profiles (i.e., amount and duration needed), depending on the site's configuration. Site capacity is specified in abstract time-resource units: one such unit corresponds to the volume of a request served in one time unit at a server, when given one resource unit. 

\subsection{Problem formulation}
\label{sec:problem}

Table~\ref{tab:buyerincentive} summarizes the notation. Let $M$ be the set of vehicles (bidders) and $K$ the set of commodities (service types), each type with total of $n_k^t$ available service slots at time $t$ in computing sites. Bidder $m \in M$ has at most $1$ demand for each service type $k \in K$ at $t$, denoted by $m_k^t \in \{0,1\}$. It draws its actions for each service type---whether to back off $\mathbf{\alpha}_m^t =\{\alpha^t_{m,1}, \cdots, \alpha^t_{m,|K|} \} \in \{0,1\}^{|K|}$, and which price to bid $\mathbf{b}_m^t =\{b^t_{m,1}, \cdots, b^t_{m,|K|} \} \in \mathbb R_+^{|K|}$---from a strategy. $m$'s utility is denoted by $u_m(\mathbf{b}_m^t)$. The bidding price $b_{m,k}^t$ is some unknown function of $m$'s private valuation $v_{m,k} \in \mathbb R_+$ of the service type, $b_{m,k}^t=f_m(v_{m,k})$. The competing bidders draw their actions from a joint distribution $\pi_{-m}^t$ based on $(\mathbf{p}^1,\cdots, \mathbf{p}^{t-1})$, where $\mathbf{p}^t \in \mathbb R_+^{|K|}$ is the payment vector received at the end of time $t$, its element $p_k^t$ is the $(n_k^t+1)^{\textrm{th}}$ highest bid for $k$. Bidder $m$ observes the new $\mathbf{p}^t$ as feedback. The auction repeats for $T$ periods. The goal is to maximize the long-term utility: $\mathcal{U}= \frac{1}{T} \sum\limits_{t=1}^T \sum\limits_{m \in M} u_m(\mathbf{b}_m^t), T\to \infty$.

\begin{table}[t]
 \centering
 \captionof{table}{Symbol definition}
 \label{tab:buyerincentive}
 \begin{tabular}{c l c l c l c }
 Sym. & Description & Sym. & Description\\
 \toprule
 $k \in K$ & service type/commodity & $n_k$ & $k$'s availability\\
 $i \in I$ & service request/bid & $v$ & bid value\\ 
 $m \in M$ & vehicle/bidder & $p$ & payment\\ 
 $x$ & bidding outcome & $u$ & utility\\
 $\alpha$ & backoff decision & $b$ & bidding price\\
 $c$ & lost bid penalty & $q$ & backoff cost\\
 $\beta$ & utilization & $B$ & budget \\
 $h \in H$ & resource types & $\omega_{i,h}$ & $i$'s requirement of $h$\\
 $Q$ & service deadline & $\rho$ & service request details \\
 $\mathbf e_m$ & $m$'s env. variables & $\text{rl}_m^t$ & $m$'s present state for RL\\
 $\text{sl}_m^t$ & $m$'s present state for SL & $P_{-m}^t$ & other bidders' state at $t$\\
 $\mathbf{a}$ & action, $\mathbf{a}=(\alpha, b)$ & $S_m^t$ & complete state for RL\\
 $\theta$ & actor parameters & $\mathbf w$ & critic parameters\\

 \bottomrule 
 \end{tabular}
\end{table}

For any $k$, when availability $n_k^t<\sum\limits_{m \in M} m_k^t$, there is more demand than available service slots and we call it ``high contention''. When $n_k^t \geq \sum\limits_{m \in M} m_k^t$, we call it ``low contention''. In a dynamic environment, $n_k^t$ depends on utilization at $t-1$ and existing demand at $t$. Due to noise and transmission delay in a realistic environment, this information is inaccurate and outdated when it becomes available to the ACA unit for admission control (Fig.\ref{flow}). 

Ideally, an auction is incentive-compatible. Unfortunately, with budget constraint and costs, the second-price auction considered here is no longer incentive-compatible. But we still use this type of auction as we can show (in Sec.\ref{lowContention} and \ref{highContention}) that it maximizes social welfare and optimally allocates resources. We also use the payment signal as additional feedback to aid the bidders' learning process (Sec.\ref{modelDescription}).

\section{Proposed Solution}
\label{sec:solution}

To solve the problem described in Sec.\ref{sec:problem}, we propose DRACO, a \textbf{D}istributed \textbf{R}einforcement-learning algorithm with \textbf{A}uction mechanism for \textbf{C}omputation \textbf{O}ffloading. In Sec.\ref{payment} we define bidder's utility function; in \ref{lowContention} and \ref{highContention}, we prove the existence of NE, maximization of welfare and Pareto optimality in the static case, under both low and high contentions. We introduce our algorithm for dynamic environment in Sec.\ref{fsp} and \ref{modelDescription}. Notations are in Table~\ref{tab:buyerincentive}. For readability, we drop notation for time $t$ in the static case.

\subsection{Utility function}
\label{payment}

In this section, we first build up the utility function based on the payoff of classic second-price auction. Then we add costs for backoff and losing the bid, incentivizing tradeoff between higher chance of success and lower communication overhead. Finally, we add the system resource utilization goal to the individual utility. 

In each auction round, if a bid $i$ for service type $k$ is admitted, its economic gain is $(v_{i,k}-p_{i,k})$. Each bidder has a given $v_{i,k}$ that is \begin{inparaenum}[1)] \item linear to the bidder's estimated resource needs for $k$ and \item within the budget. \end{inparaenum} The first condition guarantees Pareto optimality (Corollary \ref{pareto}); the second avoids overbidding under rationality (Theorem \ref{thm:spa}). Our study does not consider irrational or malicious agents, e.g., whose goal is to reduce social welfare even if individual outcome may be hurt. ACA records $b_{j,k}$ of the highest losing bid $j$ for each $k$, which is also the $(n_k+1)$th highest bidding price. For $n_k=1$ this would be the second highest price, hence the name ``second-price auction''. If $i$ is admitted, the payment $p_{i,k}=b_{j,k}$ is signaled back to the bidder. If $i$ is rejected, it has a constant cost of $c_{i,k}$. The bidder's utility so far:

\begin{flalign}\label{eq:uik}
\mathcal{u}_{i,k} =x_{i,k} \cdot (v_{i,k}-p_{i,k})-(1-x_{i,k}) \cdot c_{i,k} 
\end{flalign}

\noindent $x_{i,k}=1$ means bidder wins bid $i$ for a service slot of service type $k$, which implies $b_{i,k}$ is among the highest $n_k$ bids for $k$. Ties are broken randomly.

We add $\alpha_{i,k} \in \{0,1\}$ for backoff decision: bidder submits the bid if $\alpha_{i,k}=1$, otherwise, it backs off with a cost $q_{i,k}$:

\begin{flalign}\label{eq:rewardbackoff}
u_{i,k} = \alpha_{i,k} \cdot (\mathcal{u}_{i,k} -\mathbb{1}_{p_{i,k}=0} \cdot v_{i,k}) + (1-\alpha_{i,k}) \cdot q_{i,k} 
\end{flalign}

Especially in high contention, more rebidding causes communication overhead, but less rebidding reduces the chance of success. With $c_{i,k}$, the utility incentivizes less rebidding to reduce system-wide communication overhead (\textbf{C1}). Together with $q_{i,k}$, the bidder is incentivized to trade off between long backoff time and risky bidding. In our implementation (Sec.\ref{sec:eval}), $\alpha$ is continuous between $0$ and $1$ that also indicates the length of backoff time. 

To further align the bidder's objective with system overall objective (\textbf{C1}), we include system resource utilization $\beta$ in the utility. This incentivizes bidders to minimize system utilization. Hence, the complete utility definition is:

\begin{flalign}\label{eq:reward1}
u_i = \sum\limits_{k \in K} u_{i,k} + W \cdot (1-\beta) 
\end{flalign}

$W$ is a constant that weighs the utilization objective. In low contention, there is adequate resource to accept all bids, bidding price is less relevant, and backoff decision becomes more important.

To calculate Eq.\ref{eq:reward1}, the bidder needs only these feedback signals: bidding outcome $x_{i,k}$, final price $p_{i,k}$ and system utilization $\beta$, addressing \textbf{C2}.

\subsection{Low contention}
\label{lowContention}

Low contention is much more common in networking and presumably also in future V2X applications as abundant resources are often available. We show that in low contention, the interaction mechanism is a potential game with NE. We use the concept of potential functions to do so \cite{monderer1996potential}:

\begin{defi} $G(I,A,u)$ is an exact potential game if and only if there exists a potential function $\phi(A): A \to \mathbb{R}$ s.t. $\forall i \in I$, $u_i(b_i,b_{-i})-u_i(b'_i,b_{-i})=\phi_i(b_i,b_{-i})-\phi_i(b'_i,b_{-i}), b \in A$.
\end{defi} 

\begin{rem}\label{potentialNE} Players in a finite potential game that jointly maximize a potential function end up in NE. \end{rem}

\begin{proof} See \cite{monderer1996potential}. \end{proof}

\begin{thm} Bidders with utility as Eq.\ref{eq:reward1} participate in a game as described in Sec.\ref{sec:problem} in low contention, the game is a potential game, and the outcome is an NE.\end{thm}

\begin{proof}
In low contention, $p_{i,k}=0$, as all bids are accepted. $u_i$ is reduced to: $u_i(\alpha_i,\alpha_{-i})=\sum\limits_{k}q_{i,k}-\sum\limits_{k}\alpha_{i,k} q_{i,k} + W \Big(1-\sum_j \alpha_j \cdot \frac{\omega_j}{C}\Big)$, where $-i$ denotes bidders other than $i$. $\omega_j \in \mathbb{R}^{|K|}$ is each bid's resource requirement, $C$ is system capacity. Thus, the auction is reduced to a potential game with discrete action space $\alpha_i \in \mathbb{R}^{|K|}$, and potential function $\phi(\alpha_i,\alpha_{-i})=\sum\limits_{j, k}q_{j,k}-\sum\limits_{j,k} \alpha_{j,k}q_{j,k} + W\Big(1-\sum_j \alpha_j \cdot \frac{\omega_j}{C} \Big), \forall i,j \in I, \forall k \in K$. 

We prove in Appendix \ref{appendix:potentialGame} that $u_i(\alpha_i,\alpha_{-i})-u_i(\alpha'_i,\alpha_{-i})=\phi(\alpha_i,\alpha_{-i})-\phi(\alpha'_i,\alpha_{-i})$, and hence it is a potential game, and bidders maximizing their utilities $u_i$ also maximize the potential function $\phi$. Since $\alpha_i \in \mathbb R^{|K|}$, it is a finite potential game. According to Remark \ref{potentialNE}, the outcome is an NE.
\end{proof}

In low contention, our computation offloading problem becomes a potential game. This enables us to use online learning algorithms such as in \cite{perkins2014game} that converge regardless of other bidders' behaviors. The NE is a local maximization of the potential function: each bidder finds a balance between its backoff cost and the incentive to reduce overall utilization. Empirical results in Sec.\ref{sec:eval} confirm that over time this results in a more balanced load.

\subsection{High contention}
\label{highContention}

In high contention, $\alpha$ is used in a repeated auction to avoid congestion and ensure better reward over time. To simplify the proofs, we consider only the time steps where $\alpha=1$ (bidder joins auction). We also take a small enough $W$, such that the last term in Eq.\ref{eq:reward1} can be omitted in high contention, to further simplify the utility function in the proof.

\begin{thm}\label{thm:spa} In a second-price auction, where bidders with utility as Eq.\ref{eq:reward1} compete for service slots as commodities in high contention, \begin{inparaenum}[1)] \item bidders' best-response is of linear form, \item the outcome is an NE and \item welfare is maximized.\end{inparaenum}
\end{thm}

\begin{proof} See Appendix \ref{appendix:SPAwithpenalty}.\end{proof} 

When bidders bid for service slots, the required resources are allocated. Theorem \ref{thm:spa} guarantees the maximization of welfare (total utility of bidders), but it does not guarantee the optimality of the resource allocation, unless the following conditions are met: if bidders' valuation of the commodity is linear to its resource requirement, and all bidders have some access to resources (fairness). 

\begin{cor}\label{pareto} In a second-price auction, where $M$ bidders with utility as Eq.\ref{eq:reward1} compete in high contention, the outcome is an optimal resource allocation, if the bidders' valuation of commodities is linear to resource requirement and all bidders have a positive probability of winning.
\end{cor}

\begin{proof} See Appendix \ref{appendix:paretoOptimal}.\end{proof} 

Our setup meets both conditions.

	\begin{figure}[t]
		\centering
		\includegraphics[width=0.9\linewidth]{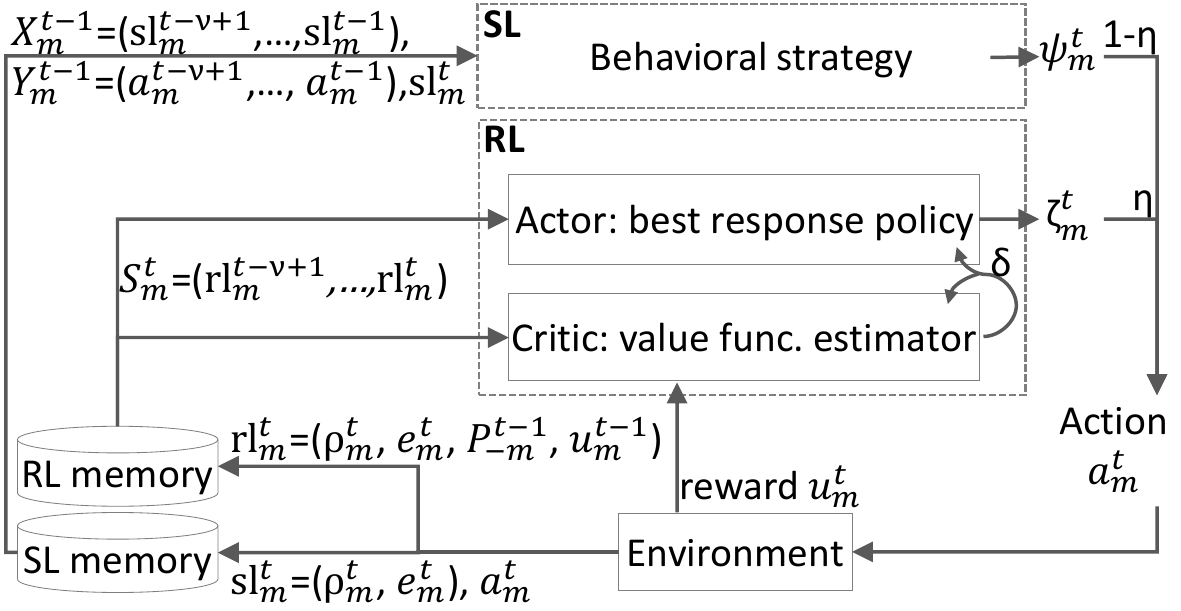}
		\vspace*{-0.2cm}
		\caption{RL and SL algorithms}
		\label{fspchart}
	\end{figure}

\subsection{The FSP algorithm}
\label{fsp}

The FSP algorithm addresses the convergence challenge of a best-response algorithm (\textbf{C3}). FSP balances exploration and exploitation by replaying its own past actions to learn an average behavioral strategy regardless of other bidders' strategies; then it cautiously plays the behavioral strategy mixed with best response \cite{heinrich2015fictitious}. The method consists of two parts: a supervised learning (SL) algorithm predicts the bidder's own behavioral strategy $\psi$, and an RL algorithm predicts its best response $\zeta$ to other bidders. The bidder has $\eta,\lim\limits_{t \to \infty} \eta =0$ probability of choosing action $\mathbf{a}=\zeta$, otherwise it chooses $\mathbf{a}=\psi$. The action includes backoff decision $\alpha$ and bidding price $b$. If $\alpha$ is above a threshold, the bidder submits the bid; otherwise, the bidder backs off for a duration linear to $\alpha$. We predefine the threshold to influence bidder behavior: with a higher threshold, the algorithm becomes more conservative and tends to back off more service requests. A learned threshold (e.g., through meta-learning algorithms) is left to future work.

Although FSP is only convergent in certain classes of games \cite{LESLIE2006285}, and in our case of a multi-player, general-sum game with infinite strategies, it does not necessarily converge to an NE, it is still an important experiment as our application belongs to a very general class of games; and empirical results show that by applying FSP, overall performance is greatly improved compared to using only RL. The FSP is described in Alg.\ref{algorithm}. 

Input to SL includes bidder $m$'s service requests---service type, resource amount required, and deadline: $\rho_m^t=\{(k_i, \omega_{i,h},Q_{i})| i \in I, h \in H\}$ ($m$ can create multiple bids, each an independent request for service type $k_i$; $\rho_m^t$ is the set of all $m$'s bids at $t$), and current environment information visible to $m$, denoted $e_m^{t}$ (e.g., number of bidders in the network and system utilization $\beta^t$). SL infers behavioral strategy $\psi_m^t$. The input $\text{sl}_m^t=(\rho_m^t,e_m^t)$ and actual action $\mathbf{a}_m^t$ are stored in SL memory to train the regression model. we use a multilayer perceptron in our implementation. 

Input to RL is constructed from $m$'s present state $\text{rl}_m^t$. $\text{rl}_m^t$ includes \begin{inparaenum}[1)] \item $\rho_m^t$; \item $e_m^t$; \item previous other bidders' state $P_{-m}^{t-1}$, represented by the final price $p_k$, or $P_{-m}^t=\mathbf{p}^t=\{ p_k^t| k \in K \}$; and \item calculated utility $u_m^{t-1}$ according to Eq.\ref{eq:reward1}. \end{inparaenum} To consider historical records, we take $\nu$ most current states to form the complete state input to RL: $S_m^t=\{\text{rl}_m^\tau|\tau=t-\nu+1,\cdots,t\}$. RL outputs best response $\zeta_m$ (Fig.\ref{fspchart}). The input consists of bidder's private information and easily obtainable public information, e.g., environment data and past prices, thus addressing \textbf{C2}.

	  \begin{algorithm}[t]
	  \small
	  \begin{algorithmic}[1]
	  \STATE Initialize $\psi_m,\zeta_m$ arbitrarily, $t=1,\eta=1/t,\nu,P_{-m}^{t-1}=\mathbf{0},u_m^{t-1}=0$, observe $e_m^t$, create $\text{rl}_m^t,\text{sl}_m^t$ and add to memory
	  \WHILE{true}
	    \STATE Take action $\mathbf{a}_m^t=(1-\eta)\psi_m^t+\eta \zeta_m^t$
	    \STATE Receive $P_{m}^t$, calculate $u_m^t$, observe $\rho_m^{t+1},\mathbf e_m^{t+1}$
	    \STATE Create and add state to RL memory: $\text{rl}_m^{t+1}$
	    \STATE Create and add state to SL memory: $(\text{sl}_m^{t+1},\mathbf{a}_m^t)$
	    \STATE Construct $S_m^t,S_m^{t+1}$, calculate $\zeta_m^{t+1}=\text{RL}(S_m^t,S_m^{t+1},u_m^t)$
	    \STATE Calculate $\psi_m^{t+1}=\text{SL}(\text{sl}_m^{t+1})$
	    \STATE $t \gets t+1$, $\eta \gets 1/t$
	  \ENDWHILE
	  \end{algorithmic}
	  \caption{FSP algorithm for bidder $m$}
	  \label{algorithm}
	  \end{algorithm}

\subsection{The RL algorithm}
\label{modelDescription}

Authors of \cite{khaledi2016optimal} use VCG and a learning algorithm for the bidders to adjust their bidding price based on budget and observation of other bidders. Our approach is similar in that we estimate other bidders' state $P_{-m}$ from payment information and use the estimate as basis for a policy. Also, similar to their work, payment information is only from the seller.

Our approach differs from \cite{khaledi2016optimal} in several major points. We use a continuous space for bidder states (i.e., continuous value for payments). As also mentioned in \cite{khaledi2016optimal}, a finer-grained state space yields better learning results. Moreover, we consider multiple commodity/service types, which is more realistic, and therefore has a wider range of applications. Further, we do not explicitly learn the transition probability of bidder states. Instead, we use historical states as input and directly determine the bidder's next action.

We use the actor-critic algorithm \cite{sutton2018reinforcement} for RL (Alg.\ref{algorithmRL}). The \textbf{critic} learns a state-value function $V(S)$. Parameters of the function are learned through a neural network that updates with $\mathbf w \gets \mathbf w + \gamma^w\delta \nabla \hat V(S, \mathbf w)$, where $\gamma$ is the learning rate and $\delta$ is the temporal difference (TD) error. For a continuing task with no terminal state, the average reward is used to calculate $\delta$ \cite{sutton2018reinforcement}: $\delta =u-\bar u+\hat V(S',\mathbf w) - \hat V(S,\mathbf w)$. In our case, the reward is utility $u$. We use exponential moving average (with rate $\lambda$) of past rewards as $\bar u$.


The \textbf{actor} learns the parameters of the policy $\pi$ in a multidimensional and continuous action space. Correlated backoff and bidding price values are assumed to be normally distributed: $F(\mu,\Sigma) = \frac{1}{\sqrt{|\Sigma|}} \exp(-\frac{1}{2}(\mathbf x-\mu)^T\Sigma^{-1}(\mathbf x - \mu))$. For faster calculation, instead of covariance $\Sigma$, we estimate lower triangular matrix $L$ ($LL^T=\Sigma$). Specifically, the actor model outputs the mean vector $\mu$ and the elements of $L$. Actor's final output $\mathbf{\zeta}$ is sampled from $F$ through: $\mathbf{\zeta} = \mu + L\mathbf{y}$, where $\mathbf{y}$ is an independent random variable from standard normal distribution. Update function is $\theta \gets \theta + \gamma^\theta \delta \nabla \ln \pi(\mathbf{a}|S,\theta)$. We use $\frac{\partial \ln F}{\partial \mu} =\Sigma (\mathbf x-\mu)$ and $\frac{\partial \ln F}{\partial \Sigma} = \frac{1}{2} (\Sigma(\mathbf x-\mu)(\mathbf x-\mu)^T\Sigma-\Sigma)$ for back-propagation.

The objective is to find a strategy that, given input $S_m^t$, determines $\mathbf{a}$ to maximize $\frac{1}{T-t}\mathbb{E}[\sum_{t'=t}^T u_m^{t'}]$. To implement the actor-critic RL, we use a stacked convolutional neural network (CNN) with highway \cite{srivastava2015training} structure similar to the discriminator in \cite{yu2017seqgan} for both actor and critic models. The stacked-CNN has diverse filter widths to cover different lengths of history and extract features, and it is easily parallelizable, compared to other sequential networks. Since state information is temporally correlated, such a sequential network extracts features better than multilayer perceptrons. The highway structure directs information flow by learning the weights of direct input and performing non-linear transform of the input.

In low contention, authors of \cite{perkins2014game} prove that an actor-critic \cite{sutton2018reinforcement} RL algorithm converges to NE in a potential game. In high contention, although we prove the existence of an NE in the static case, the convergence property of our algorithm in a stochastic game is not explicitly analyzed. We show it through empirical results in Sec.\ref{sec:eval}.

	  \begin{algorithm}[t]
	  \small
	  \begin{algorithmic}[1]
	  \STATE Initialize $\theta, w$ arbitrarily. Initialize $\lambda$
	  \WHILE{true}
	    \STATE Input $t$ and $S_m^t,S_m^{t+1}$ constructed from RL memory
	    \STATE Run critic and get $\hat V(S_m^{t}, \mathbf w),\hat V(S_m^{t+1},\mathbf w)$
	    \STATE Calculate $\bar u_m=\lambda \bar u_m$ and $\delta$ (utility $u$ is reward $R$)
	    \STATE Run actor and get $\mu(\theta), \Sigma(\theta)$
	    \STATE Sample $\zeta_m^{t+1}$ from $F(\mu,\Sigma)$, update $\mathbf w$ and $\theta$
	  \ENDWHILE
	  \end{algorithmic}
	  \caption{RL algorithm for bidder $m$}
	  \label{algorithmRL}
	  \end{algorithm}

\section{Evaluation}
\label{sec:eval}

We develop a Python discrete-event simulator, 
with varying number of vehicles of infinite lifespan, one MEC with ACA and edge computing site, and one remote computing site (extention to multiple ACA units and computing sites is left to future work). The edge and remote sites have different resource profiles. To imitate a realistic, noisy environment, the remote site is some distance to the ACA unit, such that data transmission would cause non-negligible delay in state information update. We also add a small, normally distributed noise to this delay, as well as to the actual resource required for a service. Each vehicle is randomly and independently initialized with a budget of ``high'' or ``low'' with 50\% probability. For the operating-side load-balancing policy, we apply state-of-the art resource-intensity-aware load-balancing (RIAL) \cite{8006307} with slight modifications. The method achieves dynamic load-balancing among computing sites through resource pricing that is correlated to the site's load, and loads are shifted to ``cheaper'' sites. Finally, we compare the performance of active agents (DRACO on the user side, RIAL on the operating side, or D+R) to passive agents (only RIAL on the operating side), as shown in Fig.\ref{flow}. 

We test our approach in two steps. First, we comprehensively study the performance of active agents in a synthetic setup with randomized inputs and a wide range of environment parameters. Then, we choose a realistic scenario, a 4-way traffic intersection with realistic mobility model for vehicles and with incoming service requests modeled after specific V2X applications, to show the generalization properties of DRACO.
We evaluate the following metrics:
\begin{itemize}
\item \textbf{Offloading failure rate (OFR):} Ratio of failed offloading requests rejected by ACA or not executed before deadline. 
\item \textbf{Resource utilization:} Ratio of resources effectively utilized at computing sites.
\item \textbf{Rebidding overhead:} If a bid is rejected before deadline, the vehicle can bid again. More rebidding causes communication overhead, but less rebidding reduces the chance of success. We study this tradeoff, comparing the average number of actual rebiddings per vehicle within maximum permitted-rebidding (MP).
\end{itemize}

	\begin{figure}[t]
		\centering
		\subcaptionbox{DRACO reduces the overall offloading failure rate.\label{trainingutilization-left}}{\includegraphics[width=0.5\linewidth]{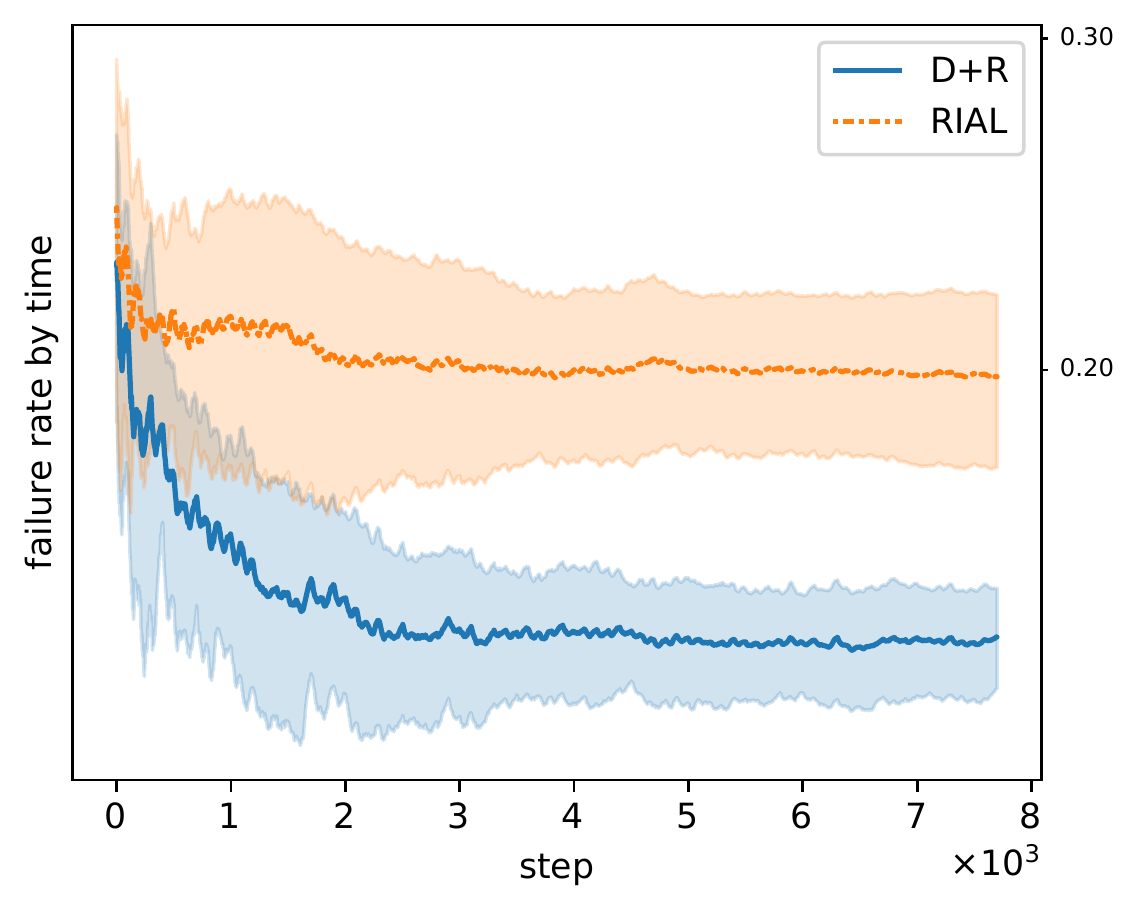}}\hfill
		\subcaptionbox{DRACO learns to better utilize resource in remote computing site.\label{trainingutilization-right}}{\includegraphics[width=0.5\linewidth]{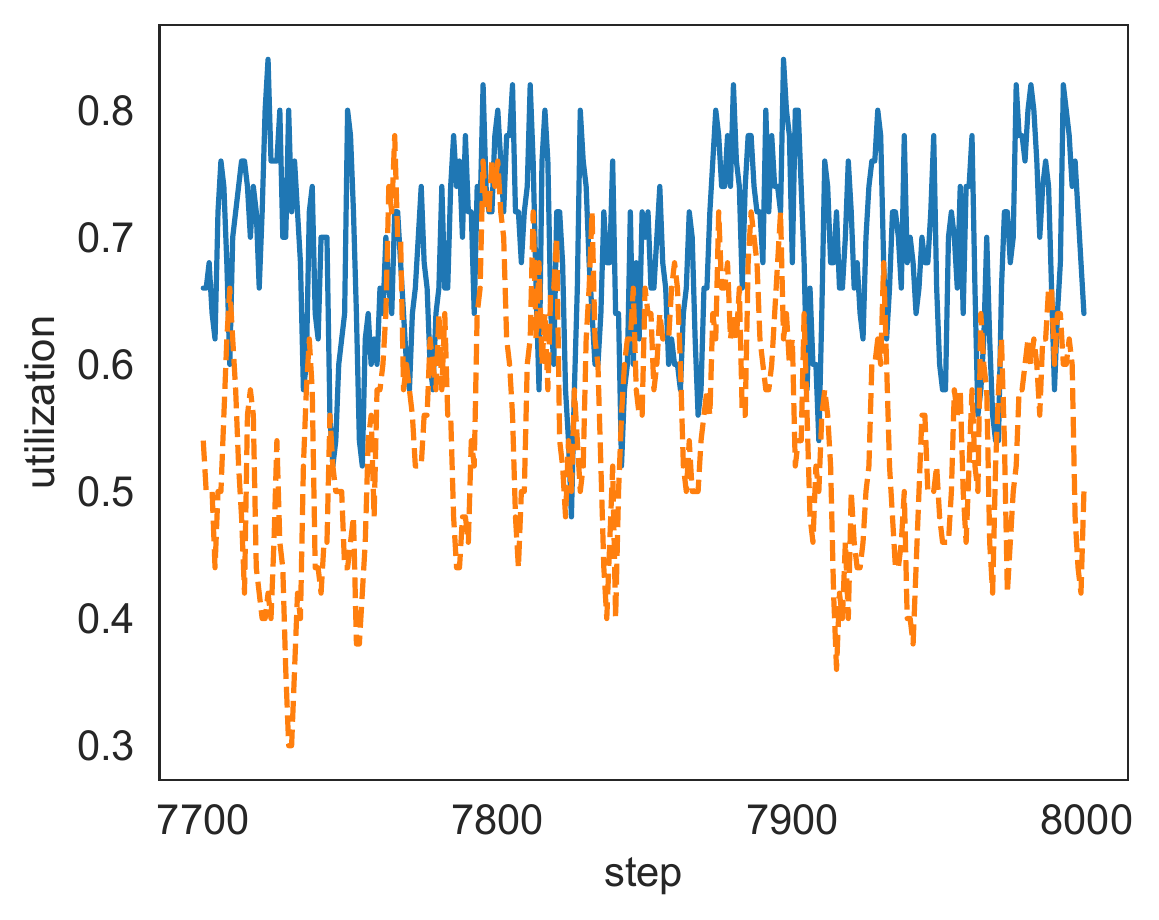}}\hfill
		\vspace*{-0.2cm}
		\caption{OFR and resource utilization, capacity=$60$, MP=$1$}
		\label{trainingutilization} 
	\end{figure}

\begin{figure*}[t]
	\centering
	\subcaptionbox{D+R reduces OFR by $40\%$, achieves $1\%$ OFR in low contention; RIAL OFR only reaches $2\%$.\label{successbycapa}}{\includegraphics[width=0.235\linewidth]{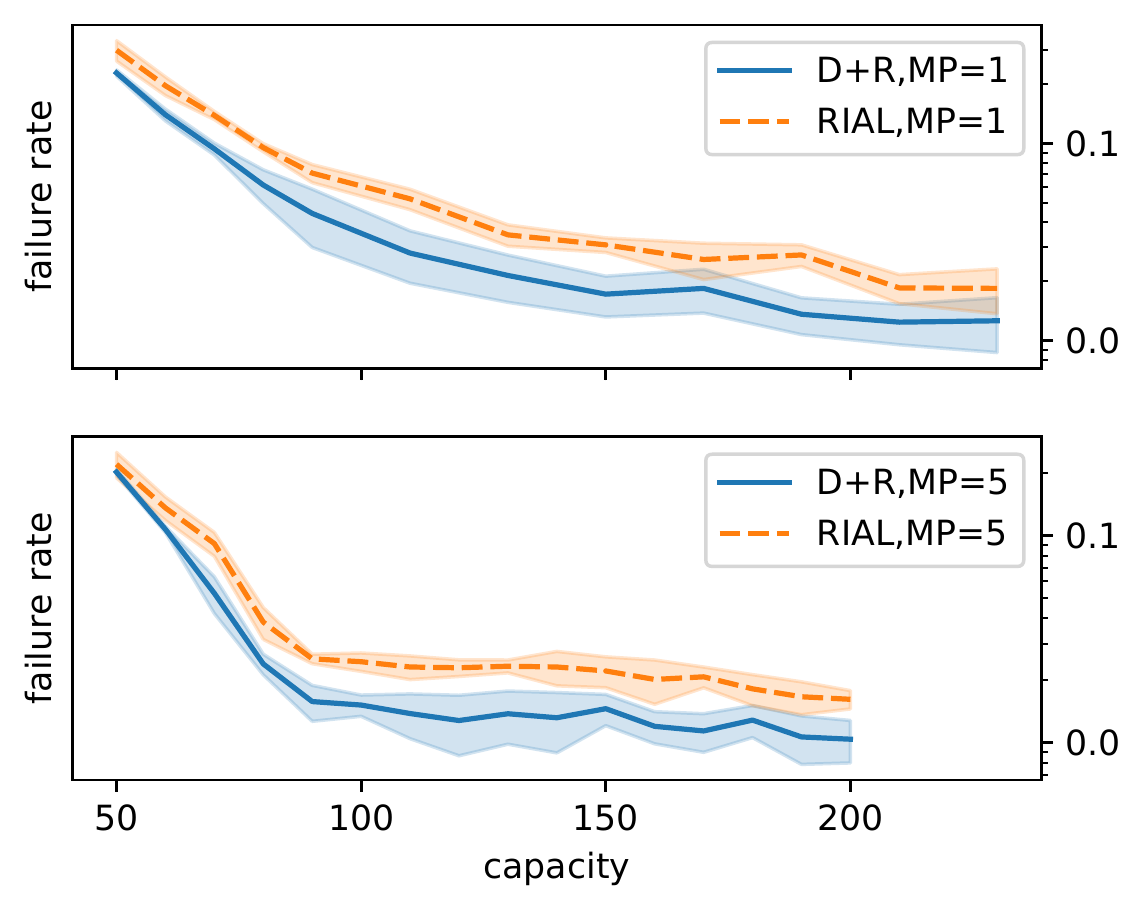}}\hfill
	\subcaptionbox{D+R needs less resource for same OFR (e.g., $2\%$ failure and MP=$1$, $38\%$ less resource needed).\label{capause}}{\includegraphics[width=0.235\linewidth]{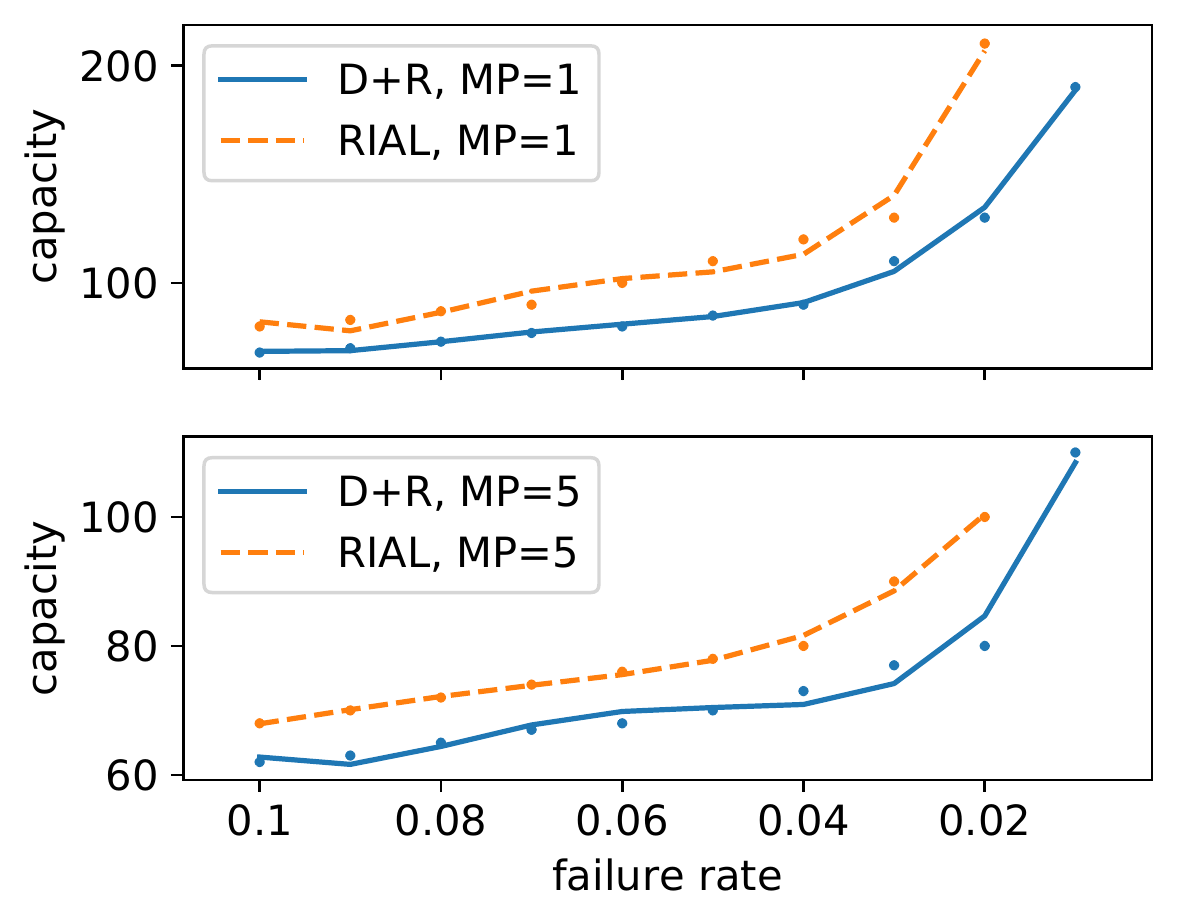}}\hfill
	\subcaptionbox{Rebidding overhead vs capacity, MP=$5$. Overhead reduces by $32\%$ on average. \label{rebiddingbycapa}}{\includegraphics[width=0.265\linewidth]{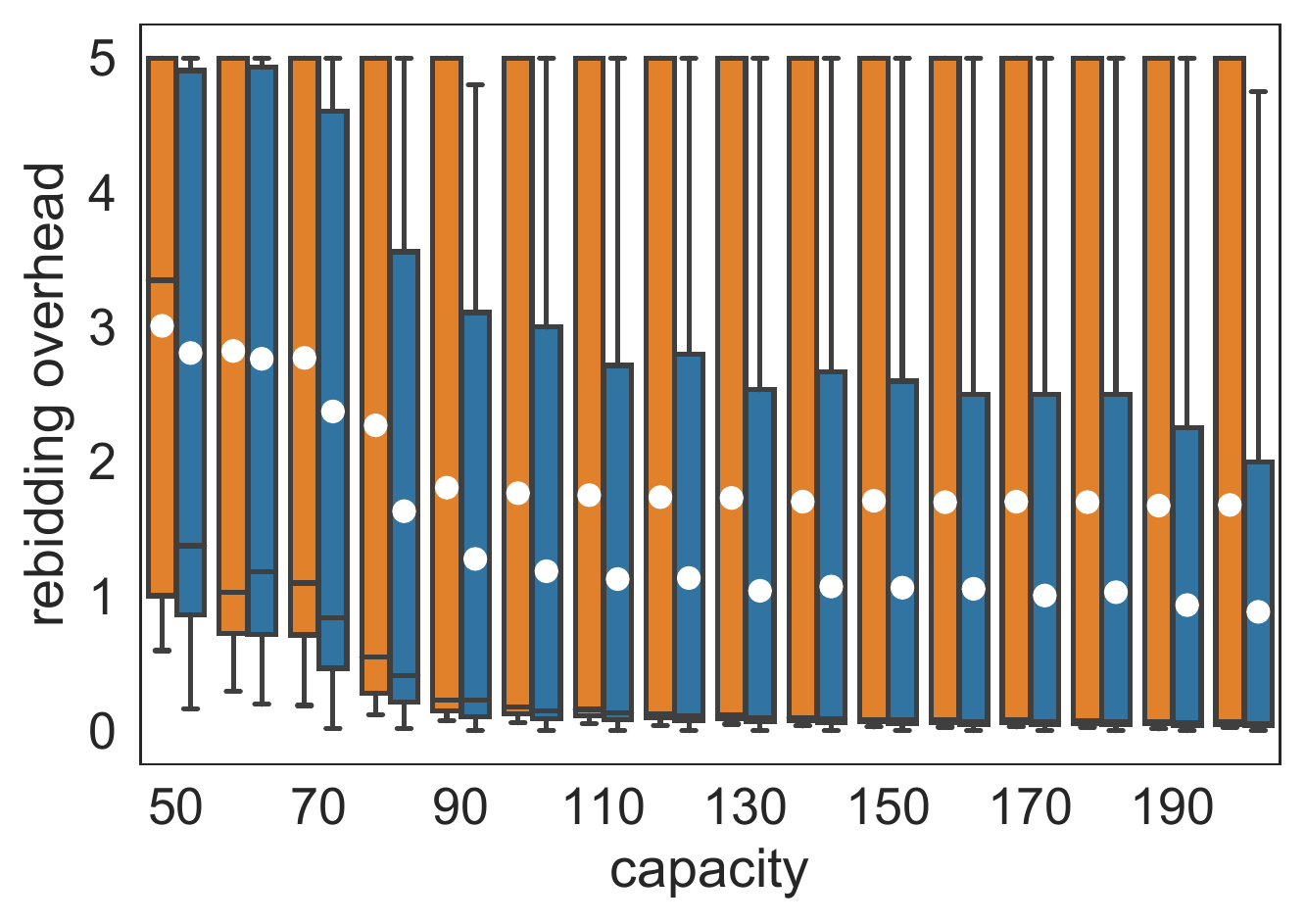}}\hfill
	\subcaptionbox{Resource utilization, MP=$1$. DRACO better utilizes resource by $18\%$ in high contention.\label{utilizationbycapa}}{\includegraphics[width=0.265\linewidth]{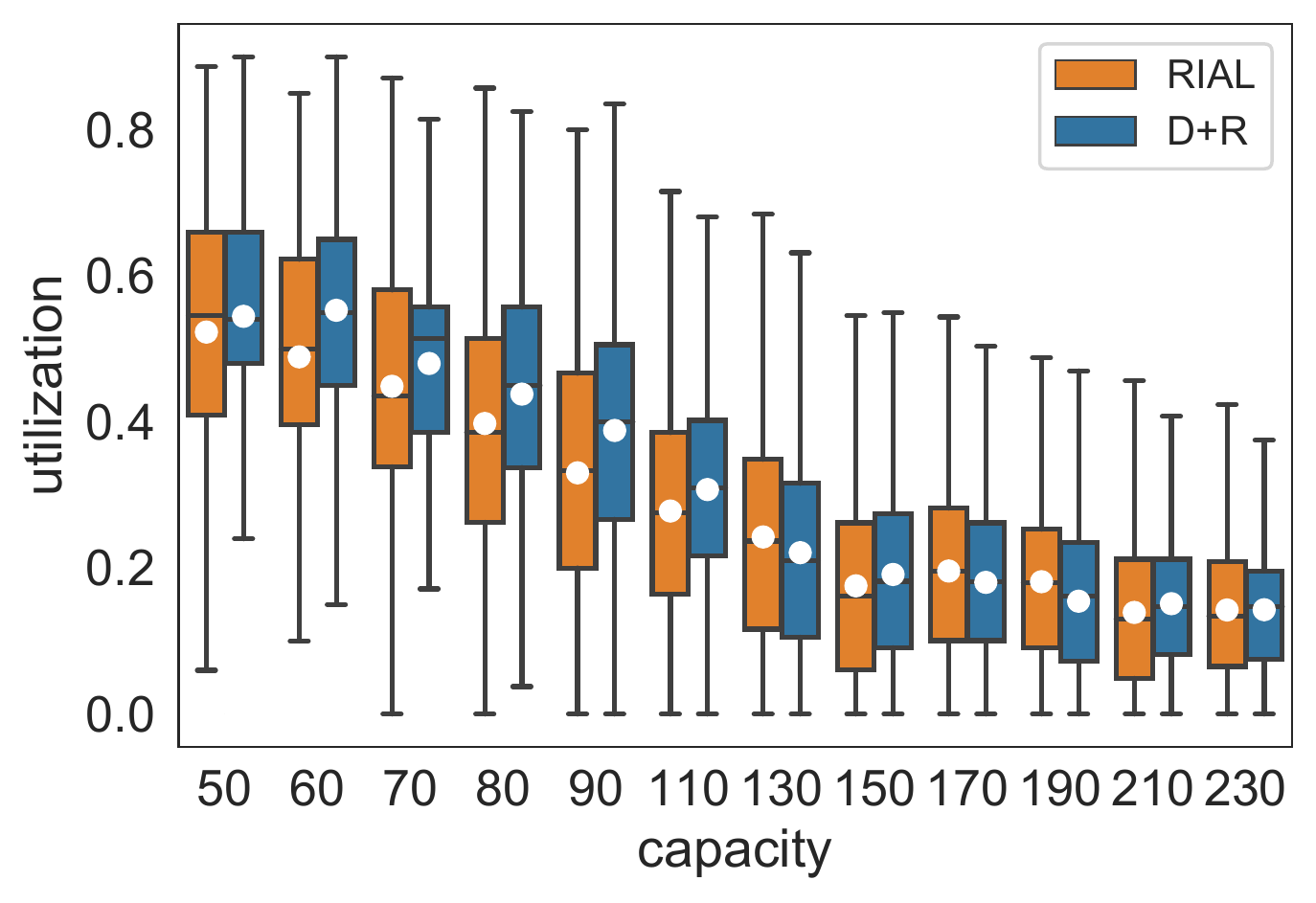}}
	\vspace*{-0.2cm}
	\caption{(a): OFR vs capacity, (b): required capacity to reach OFR$\leq10\%$, (c): rebidding overhead, (d): utilization by capacity}
	\vspace*{-0.8cm}
	\label{performancebycapa}
\end{figure*}

\subsection{Synthetic setup}
\label{subsec:eval_hypo}

In this setup, we cover a wide range of hypothetical scenarios by varying parameters such as system capacity, service/task types and number of rebidding: 
\begin{itemize}
\item Task types by resource needs in time-resource units: F1: 3 units, and F2: 30 units.
\item Service types by deadline and probability: F1, $300$ms: $18.75\%$; F1, $50$ms: $18.75\%$; F2, $300$ms: $6.25\%$; F2, $50$ms: $6.25\%$; F1-F2, $300$ms: $18.75\%$; F1-F2, $50$ms: $18.75\%$; F2-F1, $300$ms: $6.25\%$; F2-F1, $50$ms: $18.75\%$.
\item Service arrival rate per vehicle: randomized according to a two-state Markov modulated Poisson process (MMPP) \cite{wang2013characterizing}, with $\lambda_\text{high} \in (0.48,0.6), \lambda_\text{low} \in (0,0.12)$ and transition probabilities $p_\text{high}=p_\text{low}=0.6$.
\item Capacity: $50$-$230$ resource units.
\item Maximum permitted rebidding: $1$ or $5$ times, respectively.
\item Vehicle count: constant at $30$.
\item Vehicle arrival rate: $0$, always in the system; speed: $0$.
\item Data size: uniform random between $2.4$-$9.6$kbit.
\item Uplink and downlink latency: $0$.
\end{itemize}

Fig.\ref{trainingutilization-left} shows a training example where D+R's OFR is $14\%$ compared to RIAL's $20\%$ at the end of training, or a reduction of $30\%$. The lines are the mean OFR of serveral simulation runs, and the shaded area marks the standard deviation. Fig.\ref{trainingutilization-right} shows where the learning is most useful. We depict the remote site's resource utilization. Since the ACA unit's information of site utilization is delayed, with only RIAL, the site is either over-utilized or starved, in distinctive cycles (dotted line). When the vehicles learn with DRACO, they achieve better utilization (solid line). 

	\begin{figure}[t]
		\centering
		\includegraphics[width=\linewidth]{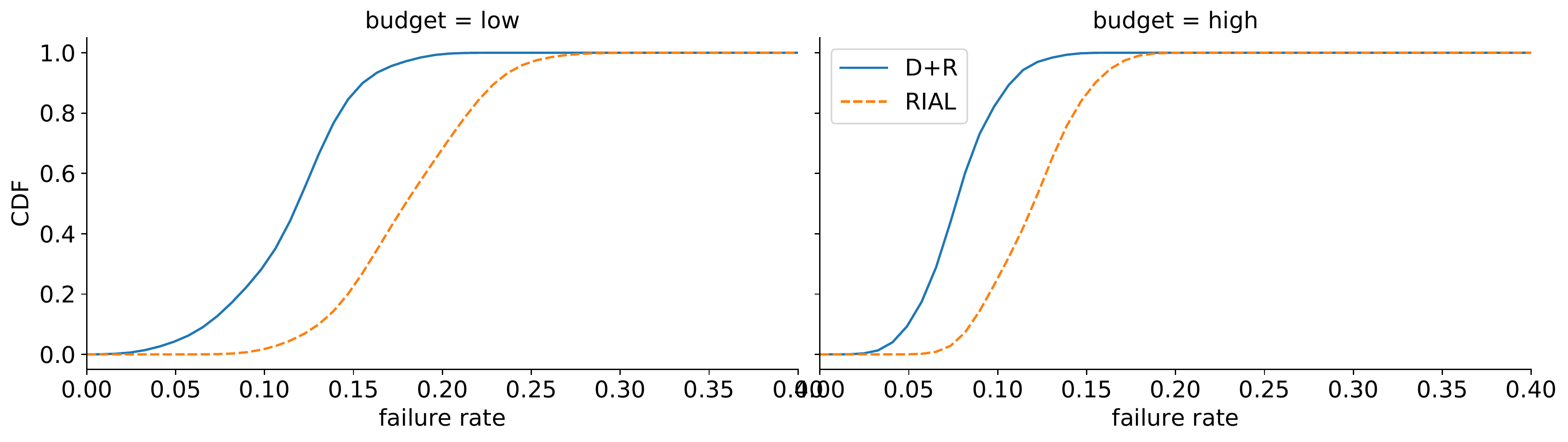}
		\vspace*{-0.2cm}
		\caption{CDF of individual OFRs (capacity=$70$, MP=$1$): DRACO reduces individual OFRs.}
		\label{cdf}
	\end{figure}

Overall OFR in all parameter settings is shown in Fig.\ref{successbycapa}. Evaluation data is collected from additional evaluation runs after the models are trained, with random incoming service requests newly generated by the MMPP. Besides requests that are not admitted by the ACA unit, the failure rate also includes requests that are admitted, but cannot be executed by the operating side before deadline (reliability). We observe that with D+R, reliability is $99\%$ and consistently higher than with RIAL for all results in the paper. We also observe that DRACO significantly reduces OFR (on average $40\%$ reduction), especially when MP is low. In low contention, i.e., capacity$\geq 100$, by efficient use of resource, D+R achieves the same level of OFR with much less resource (Fig.\ref{capause}). The improvement becomes more significant as OFR decreases. In particular, D+R reaches $1\%$ OFR with much less resource compared to RIAL regardless of MP. 

We also observe that higher MP reduces D+R's advantage over RIAL. This result is to be expected: when more rebidding is permitted, low OFR can be achieved by trial-and-error, limiting the advantage of DRACO's backoff strategy. However, trial-and-error comes with a cost: Fig.\ref{rebiddingbycapa} compares the rebidding overhead used by both algorithms when MP=$5$. In high contention, both active and passive agents leverage on rebidding, and the difference in rebidding overhead is small. D+R's advantage becomes more significant as capacity increases. The box plot shows the median (line), mean (dot), 1st to 3rd quartiles (box), and data range (whiskers). D+R imposes on average $32\%$ lower rebidding overhead. 


To validate our findings in Fig.\ref{trainingutilization-right}, we compare resource utilization under different capacities. Fig.\ref{utilizationbycapa} shows the remote site's utilization when MP=$1$. In high contention, the increase in utilization is up to $18\%$---when capacity is limited, D+R achieves lower OFR through more efficient resource usage. In low contention, capacity is less critical, active and passive agents result in similar utilization. Regardless of capacity level, D+R reduces the standard deviation in utilization by up to $21\%$.

	\begin{figure}[t]
		\centering
		\begin{minipage}{\linewidth}
		\subcaptionbox{Capacity=$50$, MP=$5$, backoff vs. price tradeoff: for all deadlines, vehicles that bid low (high) use long (short) backoff.\label{backoffprice}}{\includegraphics[width=0.5\linewidth]{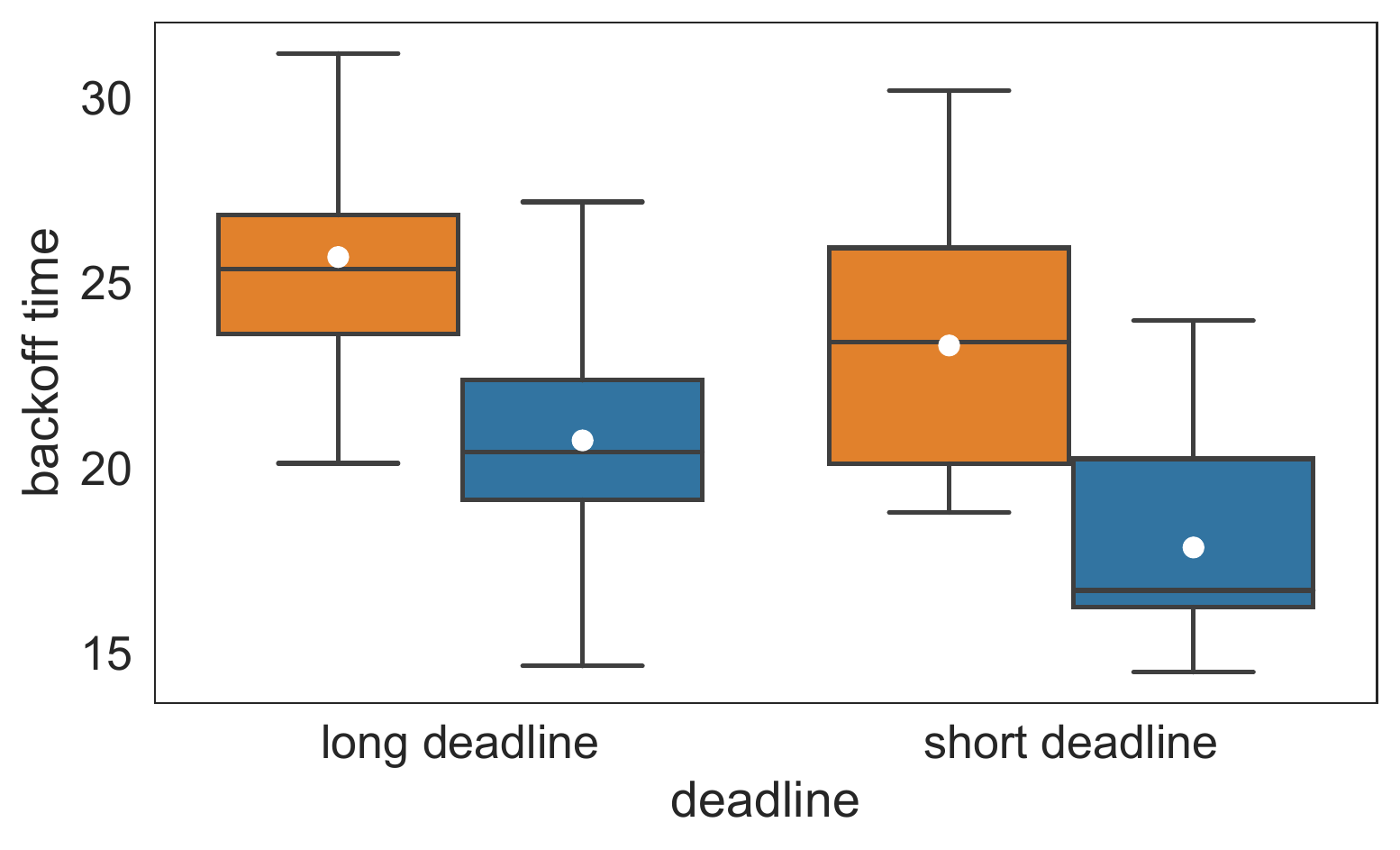}}\hfill
		\subcaptionbox{MP=$5$, long deadline, high contention: backoff time decreases with higher capacity, but the tradeoff with price remains.\label{backoffpricebycapa}}{\includegraphics[width=0.5\linewidth]{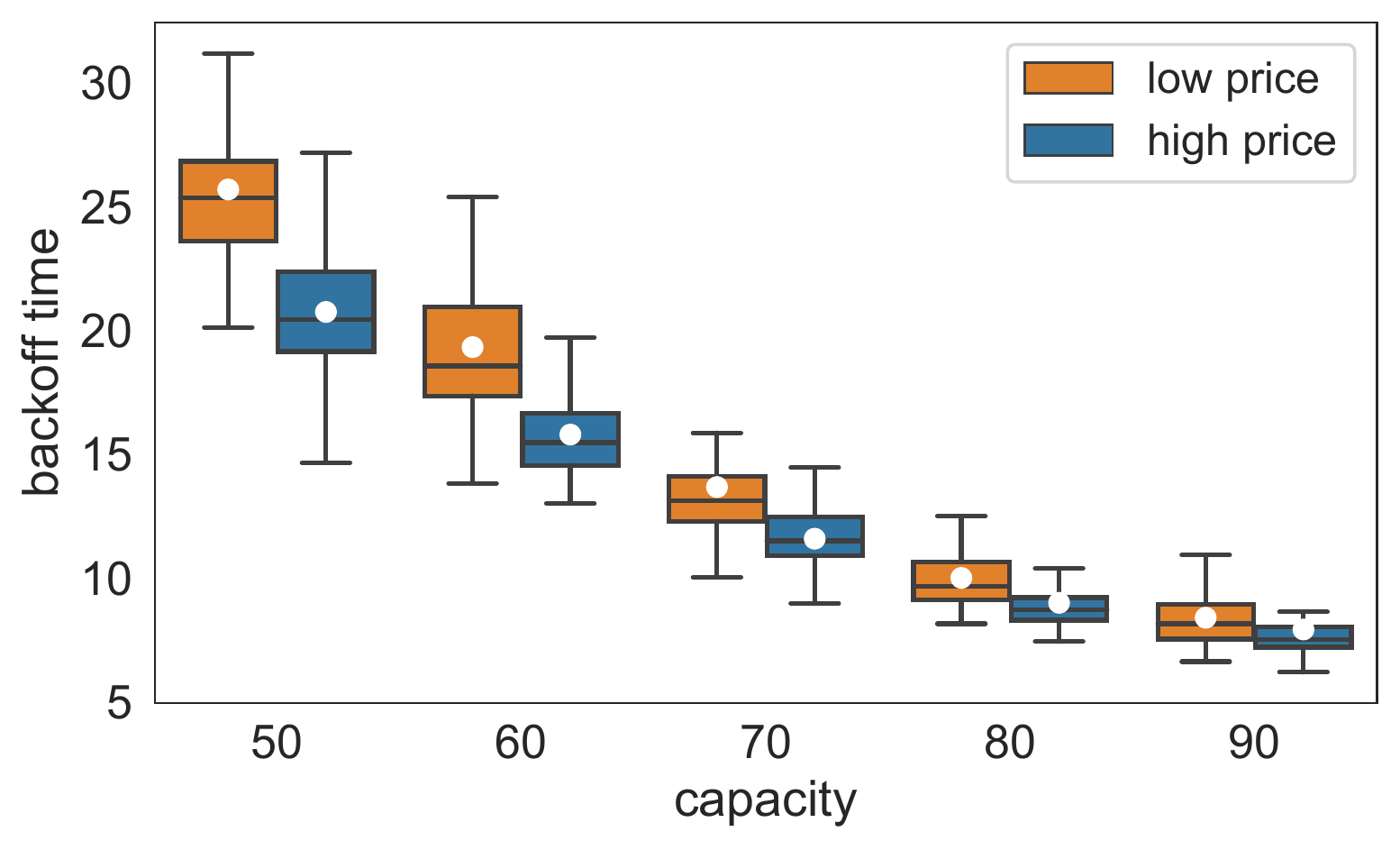}}
		\end{minipage}
		\vspace*{-0.2cm}
		\caption{Backoff and price tradeoff}
		\vspace*{-0.2cm}
		\label{backoff}
	\end{figure}

Fig.\ref{cdf} shows the cumulative probability of vehicles' individual OFRs. With DRACO, as system overall OFR reduces, the individual OFRs reduce accordingly: the auction does not cause disadvantage to individual vehicles. Moreover, vehicles with lower budget improve by a greater margin: they learn to utilize backoff mechanism to overcome their disadvantage in initial parameterization. Fig.\ref{backoffprice} shows how vehicles learn to trade off between bidding price and backoff time. They are separated into two groups: a vehicle is in the ``low price'' group if it bids on average lower than the average bidding price of all vehicles; otherwise, it is in the ``high price'' group (here we analyze actual bidding prices instead of the predefined budgets). When service requests have a longer deadline, vehicles in both price groups learn to utilize longer backoff. Regardless of the service request deadline, ``low price'' vehicles always use longer backoff in their bidding decisions, compared to the ``high price'' group. Fig.\ref{backoffpricebycapa} demonstrates the tradeoff effect with increasing capacity. As capacity increases, backoff time decreases, but the tradeoff is present in all cases. 

To summarize: Fig.\ref{trainingutilization} and \ref{performancebycapa} demonstrate DRACO's excellent overall system performance. More importantly, Fig.\ref{cdf} shows that system objective is aligned with individual objectives through incentivization (\textbf{C1}), and Fig.\ref{backoff} demonstrates where our approach fundamentally differs from previous approaches: differently initialized agents learn to select the most advantageous strategy based on limited feedback signal (\textbf{C2}). The capability to learn and behave accordingly makes our agents highly flexible in a dynamic environment. Finally, Fig.\ref{trainingutilization-left} shows good convergence speed despite computation and communication complexity of the problem (\textbf{C3}).

\subsection{Realistic setup}
\label{subsec:eval_real}

	\begin{figure*}[t]
		\centering
		\begin{minipage}{\linewidth}
		\subcaptionbox{Training env.: low contention with abundant resource, traffic phase=$10$-$40$s, low vehicle speed($10$km/h), low arrival rate=($1/2.2\text{s}$), low variation in vehicle count($22$-$30$): OFR in training(left) and evaluation(right). \label{general-train}}
			 			{  \includegraphics[width=0.22\linewidth]{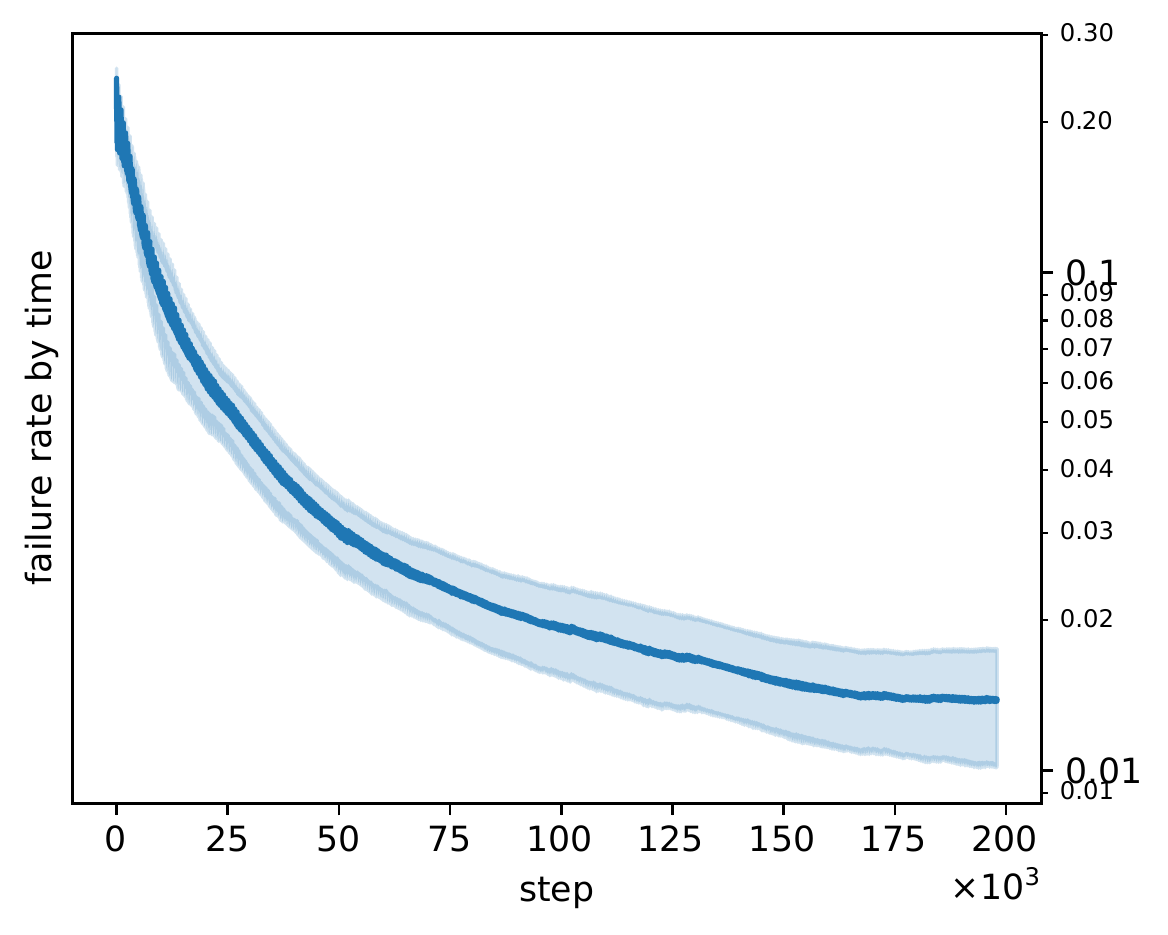}
						    \includegraphics[width=0.22\linewidth]{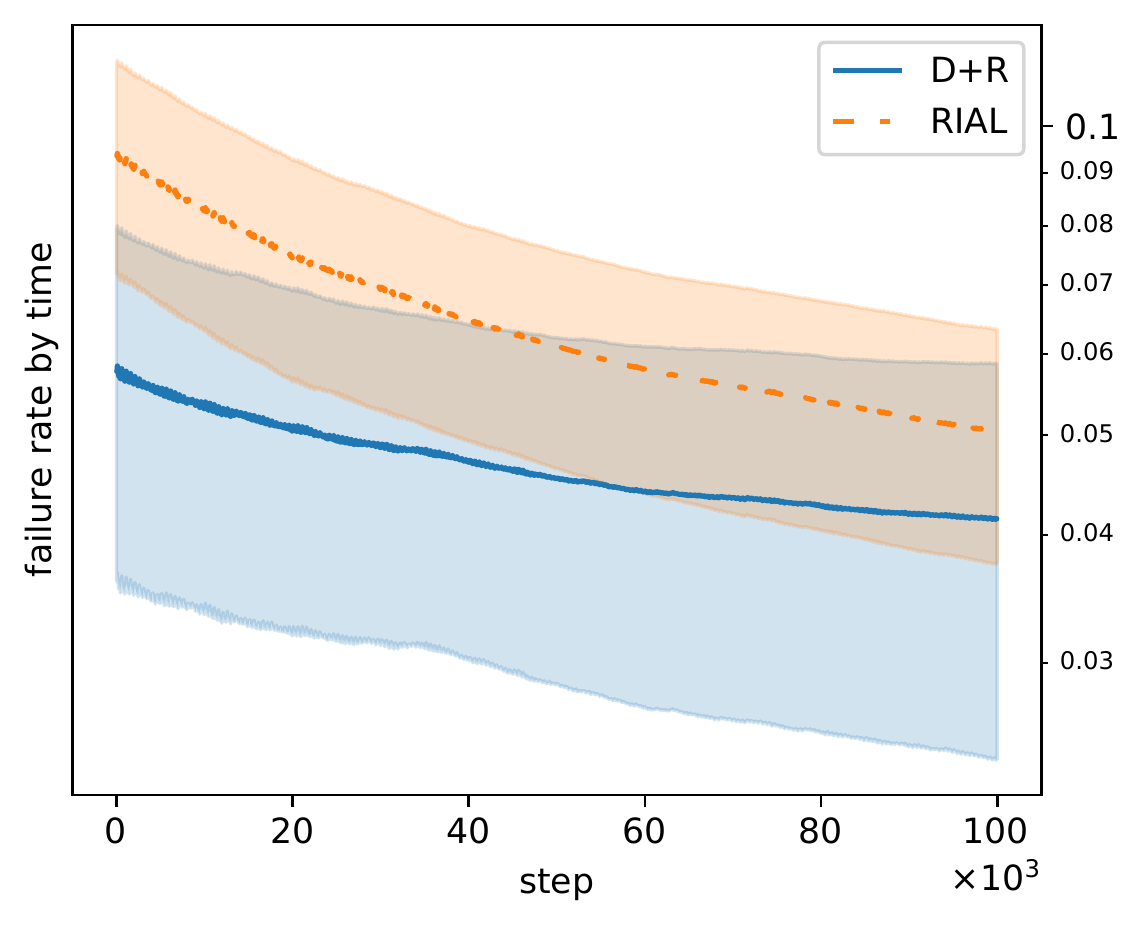}
						}\hfill
		\subcaptionbox{Test env.: high contention with limited resource, traffic phase=$20$s, high vehicle speed($30$km/h), high arrival rate($1/1\text{s}$), high variation in vehicle count($14$-$30$): vehicle count over time(left) and OFR(right). \label{general-test2}}
						{   \includegraphics[width=0.22\linewidth]{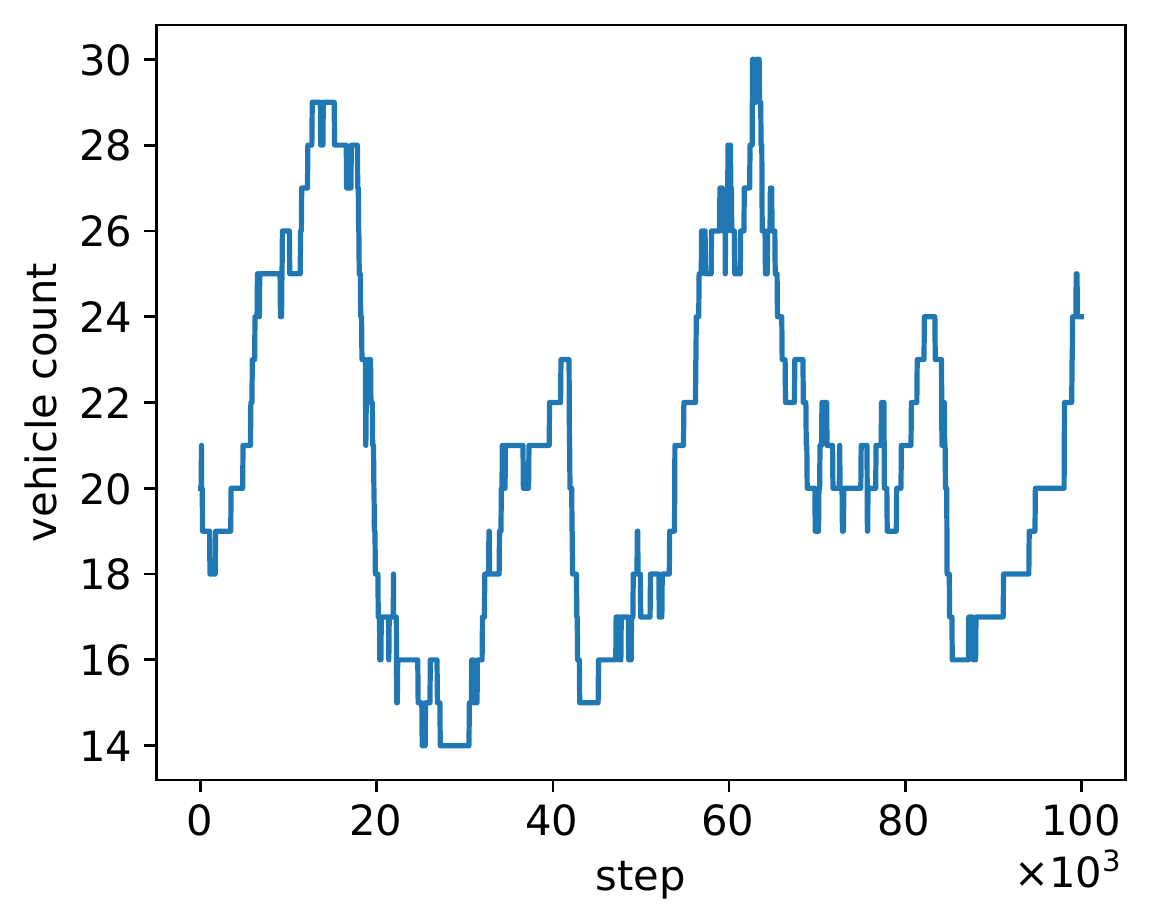}
						    \includegraphics[width=0.22\linewidth]{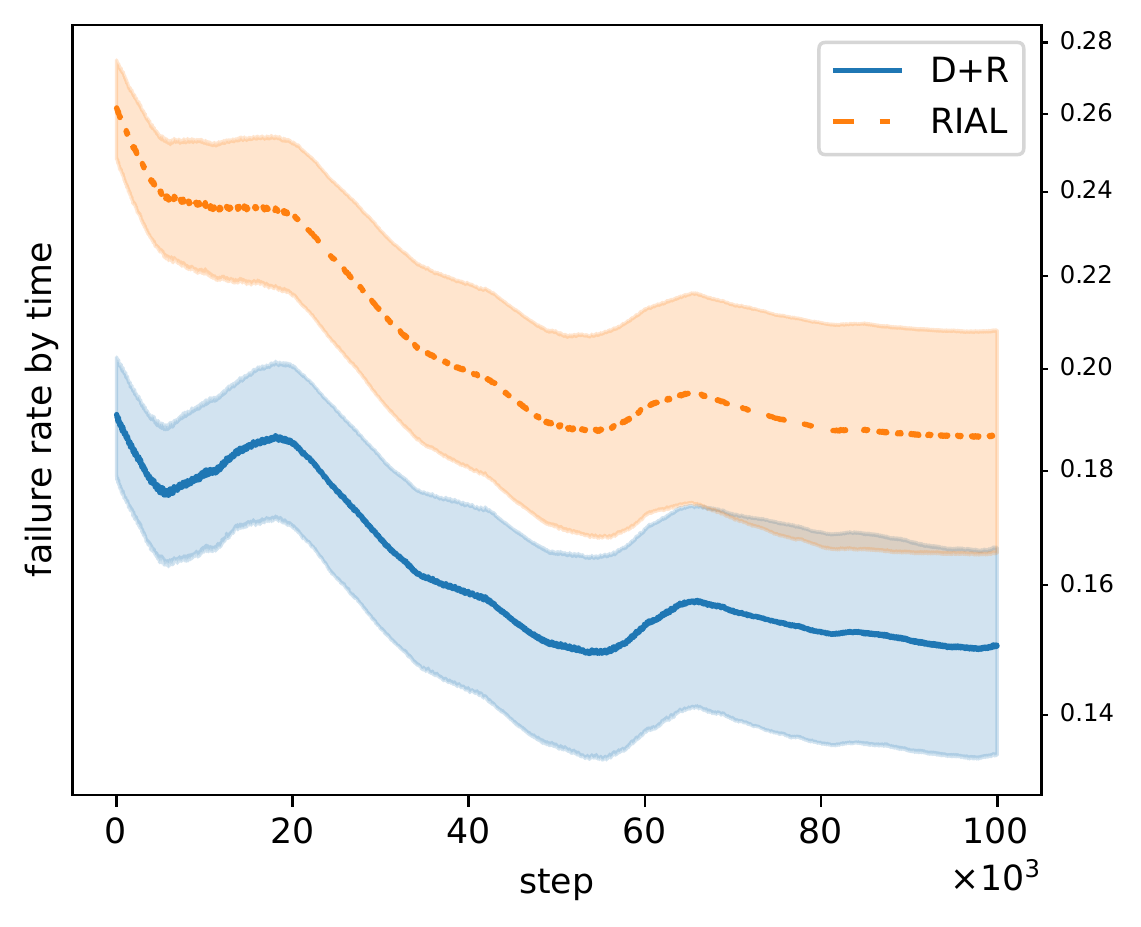}
						}\hfill
		\end{minipage}
		\vspace*{-0.2cm}
		\caption{Offloading failure rate (OFR) in training and test environments}
		\vspace*{-0.8cm}
		\label{generalization}
	\end{figure*}

In this setup, we adopt the data patterns of segmentation and motion planning applications extracted from various self-driving data projects \cite{cordts2016cityscapes} or referenced from relevant studies \cite{chen2017importance}\cite{broggi2014proud}. We also use Simulation of Urban Mobility (SUMO) \cite{behrisch2011sumo} to create a more realistic mobility model of a single junction with a centered traffic light; the junction is an area downloaded from open street map. Assuming 802.11ac protocol, we place the ACA unit in the middle of the graph and limit the edges to within 65m of the ACA. The net is with two lanes per street per direction, SUMO uniform-randomly creates a vehicle at any one of the four edges. 


Parameters of the setup are as follows \cite{cordts2016cityscapes,chen2017importance,broggi2014proud}: 
\begin{itemize}
\item Task types: F1: $80$ units, and F2: $80$ units.
\item Service types and deadline: F1: $100$ms and F2: $500$ms. 
\item Service arrival rate per vehicle: fixed at F1: every $100$ms, and F2: every $500$ms. 
\item Capacity: $20$ in high contention, $30$ in low contention.
\item Maximum permitted rebidding: $1$.
\item Vehicle count: $14$-$30$ from simulated trace data.
\item Vehicle arrival rate: constantly at $1$ every $1$ or $2.2$ seconds; speed: $10$ or $30$ km/h when driving.
\item Data size: uplink: F1: $0.4$Mbit, F2: $4$Mbit. Downlink: F1: $0$ (negligible), F2: $0.4$Mbit.
\item Latency: we take 802.11ac protocol that covers a radius of 65 meters, and assume maximum channel width of ca. $1.69$ Gbps. We model the throughput as a function of distance to the ACA unit: throughput=$-26 \times \text{distance} + 1690$ Mbps \cite{shah2015throughput}. If there are $N$ vehicles transmitting data to the ACA unit, we assume that each gets $1/N$ of the maximum throughput at that distance.
\end{itemize}

For training, we set the traffic light phases to $10$-$40$s of green for each direction, alternatively. We train our active agents with DRACO in low contention. Fig.\ref{general-train}-left shows convergence to OFR of 2\%. Then we evaluate the trained models in the same environment with newly simulated trace data from SUMO, our approach still reaches OFR of $4$\% and outperforms RIAL (Fig.\ref{general-train}-right). All simulations are repeated several times to take randomness into account.

Finally, we test our trained models in a significantly different environment, changing traffic light phases, vehicle arrival rate and speed to make the environment more volatile and dynamic, and reducing capacity to create a high-contention situation. The resulting vehicle count over time (Fig.\ref{general-test2}-left) shows a much heavier and more frequent fluctuation compared to the training environment. Note that vehicle count and OFR do not vary synchronously---OFR is determined by vehicle count and numerous other complicating factors such as transmission, queueing and processing time, past utilization, etc. Despite the significant changes to the environment, D+R still outperforms RIAL, reaching low OFRs in high contention without requiring any further training (Fig.\ref{general-test2}-right). It shows that DRACO has very good generalization properties---in fact, in the more volatile and dynamic environment, the superiority of active agents becomes more obvious.

\section{Related Work}
\label{sec:related}

Centralized approaches such as \cite{kuo2018deploying, agarwal2018joint} for resource allocation, and \cite{lyu2016multiuser, choo2018optimal, chen2018task} for offloading, are suited to core-network and data-center applications, when powerful central ACA can be set up, and data can be relatively easily obtained. They are not the focus of our study. 

Previous studies of decentralized systems address some of the issues in centralized approaches. \cite{blocher2020letting, stefan2021tnsm} propose distributed runtime algorithm to optimize system goals, but disregard user preferences. \cite{kumar2014bayesian, kumar2015coalition,chen2015efficient} only consider cooperative resource-sharing or offloading. \cite{chen2017mechanism,chen2014decentralized, cardellini2016game} require complete information to compute the desired outcome. \cite{shams2014energy} only considers discrete actions. \cite{li2019learning} learns with partial information, but it reduces complexity by assuming single service type and arrival rate. Our approach also differs from \cite{khaledi2016optimal} as we consider a multi-dimensional continuous action space with multiple service types, and both cooperative and competitive behaviors.

Besides the previously mentioned decentralized learning algorithms \cite{chang2007no, weinberg2004best, bowling2002multiagent, heinrich2015fictitious} for a dynamic environment, independent learner methods~\cite{tan1993multi} 
are used to reduce modeling and computation complexity, but they fail to guarantee equilibrium \cite{yang2018mean}, and have overfitting problems \cite{lanctot2017unified}. Finally, federated learning \cite{mcmahan2017fl} is not applicable, as it provides a logically centralized learning framework.

\section{Conclusion}
\label{sec:conclusion}

Our algorithm learns how to best utilize backoff option based on its initialization parameters. As a result, the algorithm achieves significant performance gains and very good generalization properties. Our interaction mechanism aligns private and system goals without sacrificing either user autonomy or system-wide resource efficiency, despite the distributed design with limited information-sharing.

We assume there is no ``malicious'' agent with the goal to reduce social welfare or attack the system. In general, agents with heterogeneous goals is left to future research. In V2X, all devices are potential computing sites; offloading between any devices should be considered. Long-term effect of decisions---e.g.\ if unused budget can be saved for the future, is also an interesting topic. How initialization and predefined parameters affect agent behavior and the algorithm's convergence property, needs to be studied in detail.

\bibliographystyle{IEEEtran}
\appendices
\section{Proof of potential game}
\label{appendix:potentialGame}

\begin{proof} 

We define player $i$'s utility as $u_i(\alpha_i,\alpha_{-i})=\sum\limits_{k \in K}q_{i,k}-\sum\limits_{k \in K}\alpha_{i,k} q_{i,k} + W \Big(1-\frac{\sum_j \alpha_j \cdot \omega_j}{C}\Big)$, where $\omega_j \in \mathbb{R}^K$ is the resource requirement of each commodity, $C$ is the system capacity. 

We define potential function: $\phi(\alpha_i,\alpha_{-i})=\sum\limits_{j \in I, k \in K}q_{j,k}-\sum\limits_{j \in I, k \in K} \alpha_{j,k}q_{j,k}+W\Big(1-\frac{\sum_j \alpha_j \cdot \omega_j}{C} \Big)$. 

To simplify, we substitute with $Q_i=\sum\limits_{k\in K}q_{i,k}$, $A_i=\sum\limits_{k \in K}\alpha_{i,k} q_{i,k}$, $A_{-i} = \sum\limits_{j \in I, j \neq i, k \in K}\alpha_{j,k} q_{j,k}$, $B_i=\sum\limits_k \alpha_{i,k} \omega_{i,k}$, $B_{-i}=\sum\limits_{j \in I, j \neq i, k \in K}\alpha_{j,k} \omega_{j,k}$, and rewrite: $u_i(\alpha_i,\alpha_{-i})=Q_i-A_i+W-\frac{W}{C}(B_i+B_{-i})$ and $u_i(\alpha'_i,\alpha_{-i})=Q_i-A'_i+W-\frac{W}{C}(B'_i+B_{-i})$; hence, $\phi(\alpha_i,\alpha_{-i}) = \sum\limits_j Q_j-(A_i+A_{-i})+W-\frac{W(B_i+B_{-i})}{C}$, $\phi(\alpha'_i,\alpha_{-i}) = \sum\limits_j Q_j-(A'_i+A_{-i})+W-\frac{W(B'_i+B_{-i}) }{C}$, which implies $u_i(\alpha_i,\alpha_{-i})-u_i(\alpha'_i,\alpha_{-i}) =-(A_i-A'_i)-\frac{W}{C}(B_i-B'_i) =\phi(\alpha_i,\alpha_{-i})-\phi(\alpha'_i,\alpha_{-i})$. 
\end{proof}

\section{Second-price auction}
\label{appendix:SPAwithpenalty}
Under high contention, as defined in Sec.\ref{payment}, $u_i$ is reduced to: \begin{flalign}\label{eq:ui}
u_i= \sum\limits_{k \in K} \Big(x_{i,k} \cdot (v_{i,k}-p_{i,k})-(1-x_{i,k}) \cdot c_{i,k} \Big)
\end{flalign}

We prove the theorem for $|M|=2$ and $|K|=1$, extension to other settings is straightforward. Our proof is an extension from \cite{sun2006wireless}. Unlike \cite{sun2006wireless}, we include in utility the second-price payment and cost for losing a bid. Based on \cite{sun2006wireless}, it can also be extended to multiple bidders. 

$2$ bidders receive continuously distributed valuations $v_i \in [l_i,m_i], i \in\{1,2\}$ for $1$ commodity, and choose their strategies $f_1(v_1),f_2(v_2)$ from the strategy sets $F_1$ and $F_2$. The resulting NE strategy pair is $(f_1^*, f_2^*)$. Any strategy function $f(v)$ is increasing in $v$, with $f_1(l_1)=a$, and $f_1(m_1)=b$. We assume users have budgets $(B_1,B_2)$, and that they cannot bid more than the budget. We define cost for losing the bid $c_i$.

We formulate the problem into a utility maximization problem: $\max\limits_{f_2 \in S_2(f_1)} u_2(f_1,f_2)$. We say $f_2$ is a best response of bidder 2, if $u_2(f_1,f_2)\geq u_2(f_1,f_2')$, $\forall f_2' \in S_2(f_1)$. A NE strategy pair $(f_1^*, f_2^*)$ has the strategies as each other's best responses. 

\begin{thm}\label{thm:bestResp} Given bidder 1's bidding strategy $f_1 \in F_1,f_1(l_1)=a_1,f_1(m_1)=b_1$, bidder 2's best response has the form $\begin{cases}
f_2(v_2) \leq a_1 & \text{for } v_2 \in [l_2,\theta_1] \\
f_2(v_2) = j_2 \cdot v_2 + d_2& \text{for } v_2 \in [\theta_1,\theta_2] \\
f_2(v_2) \geq b_1 & \text{for } v_2 \in [\theta_2,m_2]
\end{cases}$, where $\theta_1, \theta_2 \in [l_2,m_2]$ and $j_2 \theta_1 + d_2 =a_1, j_2\theta_2 + d_2 =b_1$.
\end{thm}

Theorem \ref{thm:bestResp} implies that the best response of bidder $1$ and $2$ are both of the linear form. Using the new best responses function, we similarly extend the proof of the NE outcome and welfare maximization to suit our case. Detailed proof is provided in supplemental meterial \cite{dracosource}.

\section{Pareto optimality}
\label{appendix:paretoOptimal}

Valuation of the service request is a linear function of the resource needed: $v_1=g_1 \omega_1 + k_1,v_2=g_2 \omega_2+k_2$, $g,k$ are constants, $\omega$ is amount of resource required. The allocation rule under NE is: $A^*_{v_1,v_2}=1 \text{, if } j_1 v_1 + d_1 \geq j_2 v_2 + d_2 \text{, otherwise } 2$. Form of the condition is from best response form in appendix Sec.\ref{appendix:SPAwithpenalty}. We also assume that both bidders have at least some access to the resources, as a form of fairness. We define the fairness constraint: $\mathbb{E}[\omega_1|_{A_{v1,v2}=1}] / \mathbb{E}[\omega_2|_{A_{v1,v2}=2}]=\gamma \in \mathbb R_{>0}$.

\begin{thm}
The allocation $A^*_{v_1,v_2}$ maximizes overall resource allocation $\omega_1+\omega_2$, subject to the fairness constraint, when the valuations are linear functions of resources. Or, the NE of the game achieves optimal resource allocation.
\end{thm}

\begin{proof}
Find the Lagrangian multiplier $\lambda^*$ that satisfies the fairness constraint with NE allocation $A^*_{v_1,v_2}$. Define $g,k$ as: $g_1 = (1+\lambda^*)/j_1 \text{ , } k_1 =-d_1/j_1$, and $g_2 = (1-\gamma \lambda^*)/j_2 \text{ , } k_2 =-d_2/j_2$. We rewrite the allocation: $A^*_{\omega_1,\omega_2} = 1 \text{, if } \omega_1 (1+\lambda^*) \geq \omega_2(1-\gamma \lambda^*) \text{, otherwise } 2$. Rest of the proof is same as \cite{sun2006wireless}.
\end{proof}

\bibliography{IEEEabrv,references}
\end{document}